\documentclass[a4paper]{article}

\usepackage[linesnumbered,ruled,vlined]{algorithm2e}
\usepackage{amsmath,amssymb,amsthm}
\newtheorem{theorem}{Theorem}[section]
\newtheorem{lemma}[theorem]{Lemma}
\newtheorem{proposition}[theorem]{Proposition}

\newtheorem{definition}[theorem]{Definition}
\newtheorem{remark}[theorem]{Remark}

\usepackage{authblk}

\usepackage{array,multirow,graphicx,booktabs}
\usepackage[colorinlistoftodos]{todonotes}

\newcommand{\ot}{\leftarrow}

\newcommand{\cF}{\mathcal{F}}
\newcommand{\cS}{\mathcal{S}}
\newcommand{\cP}{\mathcal{P}}
\newcommand{\cI}{\mathcal{I}}
\newcommand{\conv}{\mathop{\rm conv}}
\newcommand{\supp}{\mathop{\rm supp}}
\newcommand{\argmax}{\mathop{\rm arg\,max}}
\newcommand{\argmin}{\mathop{\rm arg\,min}}

\title{Randomized Strategies for Robust Combinatorial Optimization}
\date{}

\author[1]{Yasushi Kawase}
\author[2]{Hanna Sumita}
\affil[1]{Tokyo Institute of Technology, Tokyo, Japan. \texttt{kawase.y.ab@m.titech.ac.jp}}
\affil[2]{Tokyo Metropolitan University, Tokyo, Japan. \texttt{sumita@tmu.ac.jp}}

\begin{document}

\maketitle

\begin{abstract}
In this paper, we study the following robust optimization problem. 
Given an independence system and candidate objective functions, we choose an independent set, and then an adversary chooses one objective function, knowing our choice. 
Our goal is to find a randomized strategy (i.e., a probability distribution over the independent sets) that maximizes the expected objective value. 
To solve the problem, we propose two types of schemes for designing approximation algorithms.
One scheme is for the case when objective functions are linear. 
It first finds an approximately optimal aggregated strategy and then retrieves a desired solution with little loss of the objective value. 
The approximation ratio depends on a relaxation of an independence system polytope. 
As applications, we provide approximation algorithms for a knapsack constraint or a matroid intersection by developing appropriate relaxations and retrievals. 
The other scheme is based on the multiplicative weights update method. 
A key technique is to introduce a new concept called $(\eta,\gamma)$-reductions for objective functions with parameters $\eta, \gamma$. 
We show that our scheme outputs a nearly $\alpha$-approximate solution if there exists an $\alpha$-approximation algorithm for a subproblem defined by $(\eta,\gamma)$-reductions. 
This improves approximation ratio in previous results. 
Using our result, we provide approximation algorithms when the objective functions are submodular or correspond to the cardinality robustness for the knapsack problem.
\end{abstract}

\section{Introduction}
This paper addresses robust combinatorial optimization.
Let $E$ be a finite ground set, and let $n$ be a positive integer. 
Suppose that we are given $n$ set functions $f_1,\dots,f_n\colon 2^E\to\mathbb{R}_+$ and an independence system $(E,\cI)$.
The functions $f_1,\dots,f_n$ represent possible scenarios.
The worst case value for $X\in\cI$ across all scenarios is defined as $\min_{k\in[n]} f_k(X)$, where $[n]=\{1,\dots,n\}$.
We focus on a randomized strategy for the robust optimization problem,
i.e., a probability distribution over $\cI$.
Let $\Delta(\cI)$ and $\Delta_n$ denote the set of probability distributions over $\cI$ and $[n]$, respectively.
The worst case value for a randomized strategy $p \in \Delta(\cI)$ is defined as $\min_{k\in[n]}\sum_{X\in\cI}p_X\cdot f_k(X)$. 
The aim of this paper is to solve the following robust optimization problem:
\begin{align}\label{eq:problem}
\max~\min_{k \in [n]}\sum_{X\in\cI}p_X\cdot f_k(X)\quad\text{s.t.}\quad p\in\Delta(\cI).
\end{align}

There exist a lot of previous work on a deterministic strategy for \eqref{eq:problem}, that is, finding $X \in \cI$ that maximizes the worst case value. 
We are motivated by the following two merits to focus on a randomized strategy. 
The first one is that the randomization improves the worst case value dramatically.
Suppose that $f_1(X)=|X\cap\{a\}|$, $f_2(X)=|X\cap\{b\}|$, and $\cI=\{\emptyset,\{a\},\{b\}\}$.
Then, the maximum worst case value of deterministic strategy is $\max_{X\in\cI}\min_{k\in\{1,2\}} f_k(X)=0$,
while that of randomized strategy is $\max_{p\in\Delta(\cI)}\min_{k\in\{1,2\}}\sum_{X\in\cI}p_X\cdot f_k(X)=1/2$. 
The second merit is that a randomized strategy can be found more easily than a deterministic one.
It is known that finding a deterministic solution is hard even in a simple setting~\cite{ABV2009,KZ2016}.
In particular, as we will see later (Theorem~\ref{thm:dethard}), computing a solution $X$ with the maximum worst case value is NP-hard even to approximate even for linear objectives subject to a cardinality constraint. 
Note that the randomized version of this problem is polynomial-time solvable (see Theorem~\ref{thm:LP-based}). 

It is worth noting that we can regard the optimal value of \eqref{eq:problem} as the \emph{game value} in a two-person zero-sum game where one player (algorithm) selects a feasible solution $p\in \Delta(\cI)$ and the other player (adversary) selects a possible scenario $q\in \Delta_n$.

An example of the robust optimization problem appears in the (zero-sum) \emph{security games},
which models the interaction between a system \emph{defender} and a malicious \emph{attacker} to the system~\cite{tambe2011}.
The model and its game-theoretic solution have various applications in the real world:
the Los Angeles International Airport to randomize deployment of their limited security resources~\cite{PJM+2008};
the Federal Air Marshals Service to randomize the allocation of air marshals to flights~\cite{TRKOT2009};
the United States Coast Guard to recommend randomized patrolling strategies for the coast guard~\cite{APS2011};
and many other agencies.
In this game, we are given $n$ targets $E$.
The defender selects a set of targets $X\in\cI\subseteq 2^E$, and then the attacker selects a facility $e\in E$.
The utility of defender is $r_i$ if $i\in X$ and $c_i$ if $i\not\in X$.
Then, we can interpret the game as the robust optimization with $f_i(X)=c_i+\sum_{j\in X}w_{ij}$
where $w_{ij}=r_i-c_i$ if $i=j$ and $0$ if $i\ne j$ for $i,j\in E$.
Most of the literature has focused on the computation of the \emph{Stakelberg equilibrium},
which is equivalent to \eqref{eq:problem}.

Another example of \eqref{eq:problem} is to compute the \emph{cardinality robustness} for the maximum weight independent set problem~\cite{hassin2002rm,fujita2013,kakimura2013ris,matuschke2015,KT2016}.
The problem is to choose an independent set of size at most $k$ with as large total weight as possible, but the cardinality bound $k$ is not known in advance. 
For each independent set $X$, we denote the total weight of the $k$ heaviest elements in $X$ by $v_{\le k}(X)$. 
The problem is also described as the following zero-sum game. 
First, the algorithm chooses an independent set $X$, and then the adversary (or nature) chooses a cardinality bound $k$, knowing $X$. 
The payoff of the algorithm is $v_{\leq k}(X)$. 
For $\alpha \in [0,1]$, an independent set $X\in\cI$ is said to be \emph{$\alpha$-robust} if $v_{\le k}(X) \geq \alpha \cdot \max_{Y\in\cI}v_{\le k}(Y)$ for any $k\in[n]$.
Then, our goal is to find a randomized strategy that maximizes the robustness $\alpha$, i.e., $\max_{p\in\Delta(\cI)} \min_{k \in [n]} {\sum_{X\in\cI} p_X\cdot v_{\le k}(X)}/{\max_{Y\in\cI}v_{\le k}(Y)}$. 
We refer this problem as the maximum cardinality robustness problem. 
This is formulated as \eqref{eq:problem} by setting $f_k(X)=v_{\le k}(X)/\max_{Y\in\cI}v_{\le k}(Y)$.

\medskip

Since \eqref{eq:problem} can be regarded as the problem of computing the game value of the two-person zero-sum game, 
one most standard way to solve \eqref{eq:problem} is to use the \emph{linear programming} (LP).
In fact, it is known that we can compute the exact game value in polynomial time with respect to the numbers of deterministic (pure) strategies for both players~(see, e.g., \cite{nisan2007agt,bowles2009} for the detail).
However, in our setting, direct use of the LP formulation does not give an efficient algorithm, because 
the set of deterministic strategies for the algorithm is $\cI$, whose cardinality is exponentially large, and hence 
the numbers of the variables and the constraints in the LP formulation are exponentially large. 

Another known way to solve \eqref{eq:problem} is to use the multiplicative weights update (MWU) method. 
The MWU method is an algorithmic technique which maintains a distribution on a certain set of interest and updates it iteratively by multiplying the probability mass of elements by suitably chosen factors based on feedback obtained by running another algorithm on the distribution~\cite{Kale2007}. 
MWU is a simple but powerful method that is used in wide areas such as game theory, machine learning, computational geometry, optimization, and so on.
Freund and Schapire~\cite{FS1999} showed that MWU can be used to calculate the approximate value of a two-person zero-sum game under some conditions. 
More precisely, if (i) the adversary has a polynomial size deterministic strategies and (ii) the algorithm can compute a \emph{best response}, 
then MWU gives a polynomial-time algorithm to compute the game value up to an additive error of $\epsilon$ for any fixed constant $\epsilon>0$.
For each $q\in\Delta_n$, we call $X^*\in\cI$ a best response for $q$ if $X^*\in\argmax_{X\in\cI}\sum_{k\in[n]}q_k f_k(X)$. 
Krause et al.~\cite{krause2011} and Chen et al.~\cite{CLSS2017} extended this result for the case when the algorithm can only compute an \emph{$\alpha$-best response}, i.e., an $\alpha$-approximate solution for $\max_{X\in\cI}\sum_{k\in[n]}q_k f_k(X)$. 
They provided a polynomial-time algorithm that finds an $\alpha$-approximation of the game value up to additive error of $\epsilon\cdot\max_{k\in[n],~X\in\cI}f_k(X)$ for any fixed constant $\epsilon>0$. 
This implies an approximation ratio of $\alpha-\epsilon\cdot\max_{k\in[n], \, X\in\cI}f_k(X)/\nu^*$, where $\nu^*$ is the optimal value of \eqref{eq:problem}. 
Their algorithms require pseudo-polynomial time to obtain an $(\alpha-\epsilon')$-approximation solution for a fixed constant $\epsilon'>0$. 
In this paper, we improve their technique to find it in polynomial time. 

The main results of this paper are two general schemes for solving \eqref{eq:problem} based on LP and MWU in the form of using some subproblems. 
Therefore, when we want to solve a specific class of the problem \eqref{eq:problem}, it suffices to solve the subproblem. 
As consequences of our results, we show (approximation) algorithms to solve \eqref{eq:problem} in which the objective functions and the constraint belong to well-known classes in combinatorial optimization, such as submodular functions, knapsack/matroid/$\mu$-matroid intersection constraints.

\subsection*{Related work}
While there exist still few papers on randomized strategies of the robust optimization problems, 
algorithms to find a deterministic strategy have been intensively studied in various setting. 
See also survey papers~\cite{ABV2009,KZ2016}.
Krause et al.~\cite{krause2008} focused on $\max_{X\subseteq E,\,|X|\le \ell}\min_{k\in[n]}f_k(X)$ where $f_k$'s are monotone submodular functions.
Those authors showed that this problem is NP-hard even to approximate, and provided an algorithm that outputs a set $X$ of size $\ell \cdot (1+\log(\max_{e\in E}\sum_{k\in[n]}f_k(\{e\})))$ whose objective value is at least as good as the optimal value.
Orlin et al.~\cite{OSU2016} provided constant-factor approximate algorithms to solve $\max_{X\subseteq E,\,|X|\le k}\min_{Z\subseteq X,\,|Z|\le \tau} f(X-Z)$, where $f$ is a monotone submodular function. 

Kakimura et al.~\cite{Kakimura2012} proved that the deterministic version of the maximum cardinality robustness problem is weakly NP-hard but admits an FPTAS. 
Since Hassin and Rubinstein~\cite{hassin2002rm} introduced the notion of the cardinality robustness, many papers have been investigating the value of the maximum cardinality robustness~\cite{hassin2002rm,fujita2013,kakimura2013ris}. 
Matuschke et al.~\cite{matuschke2015} introduced randomized strategies for the cardinality robustness, and they presented a randomized strategy with $(1/\ln 4)$-robustness for a certain class of independence system $\cI$.
Kobayashi and Takazawa~\cite{KT2016} focused on independence systems that are defined from the knapsack problem, and exhibited two randomized strategy with robustness $\Omega(1/\log\sigma)$ and $\Omega(1/\log\upsilon)$, where $\sigma$ is the exchangeability of the independence system and $\upsilon = \frac{\text{the size of a maximum independent set}}{\text{the size of a minimum dependent set} - 1}$. 

When $n=1$, the deterministic version of the robust optimization problem \eqref{eq:problem} is exactly the classical optimization problem $\max_{X \in \cI} f(X)$. 
For the monotone submodular function maximization problem, there exist $(1-1/e)$-approximation algorithms under a knapsack constraint~\cite{Sviridenko04} or a matroid constraint~\cite{CCPV2007,FW2012}, and there exists a $1/(\mu+\epsilon)$-approximation algorithm under a $\mu$-matroid intersection constraint for any fixed $\epsilon>0$~\cite{LSV2010}. 
For the unconstrained non-monotone submodular function maximization problem, there exists a $1/2$-approximation algorithm, and this is best possible~\cite{FMVV2011,BFNS2015}. 
As for the case when the objective function $f$ is linear, 
the knapsack problem admits an FPTAS~\cite{KMS2000}. 

\subsection*{Our results}

\paragraph*{LP-based algorithm}
We focus on the case when all the objective functions $f_1, \ldots, f_n$ are linear. 
In a known LP formulation for zero-sum games, each variable corresponds a probability that each $X \in \cI$ is chosen. 
Because $|\cI|$ is large, we use another LP formulation of \eqref{eq:problem}. 
The number of variables is reduced by setting as a variable a probability that each element in $E$ is chosen. 
The feasible region consists of the independence system polytope, that is, the convex hull of the characteristic vectors for $X \in \cI$. 
Although our LP formulation still has the exponential number of constraints, we can use the result by Gr{\"o}tschel, Lov{\'a}sz, and Schrijver~\cite{GLS2012} that if we can efficiently solve the separation problem for the polytope of the feasible region, then we can efficiently solve the LP by the ellipsoid method. 
Since the solution of the LP is an aggregated strategy for \eqref{eq:problem}, we must retrieve a randomized strategy from it. 
To do this, we use the result in \cite{GLS2012} again that we can efficiently compute the convex combination of the optimal vector with extreme points (vertex) of the polytope. 
Consequently, we can see that there exists a polynomial-time algorithm for \eqref{eq:problem} when $(E, \cI)$ is a matroid (or a matroid intersection), because a matroid (intersection) polytope admits a polynomial-time separation algorithm. 
As another application, we also provide a polynomial-time algorithm for the robust shortest $s$--$t$ path problem by using the dominant of an $s$--$t$ path polytope.

Moreover, we extend our scheme to deal with the case that the separation problem is NP-hard. 
For many combinatorial optimization problems such as the knapsack problem and the $\mu$-matroid intersection problem ($\mu \geq 3$), the existence of an efficient algorithm to solve the separation problem is still unknown. 
A key point to deal such cases is to use a slight relaxation of the independence system polytope. 
We show that if we can efficiently solve the separation problem for the relaxed polytope, then we can know an approximate value of \eqref{eq:problem}.
The approximation ratio is equal to the gap between the original polytope and the relaxed polytope. 
The most difficult point is the translation of the optimal solution of the LP to a randomized strategy, because the optimal solution may not belong to the original feasible region, and we are no longer able to use the result in \cite{GLS2012}. 
Instead, we compute a randomized strategy approximately. 
We demonstrate our extended scheme for the knapsack constraint and the $\mu$-matroid intersection constraint 
by developing appropriate relaxations and retrievals for them.
As results, we obtain a PTAS and a $2/(e\mu)$-approximation algorithm for the knapsack constraint and the $\mu$-matroid intersection constraint, respectively.

The merit of the LP-based algorithm compared with MWU is that the LP-based one is applicable to the case when the set of possible objective functions is given by a half-space representation of a polytope. 
The problem \eqref{eq:problem} is equivalent to the case where the set of possible objective functions is given by a convex hull of linear functions $\conv\{f_1,\dots,f_n\}$ (i.e., a vertex representation). 
Since a vertex representation can be transformed to a half-space representation (by an extended formulation as we will describe later), \eqref{eq:problem} with a half-space representation is a generalization of the original problem. 
On the other hand, the transformation of a half-space representation to a vertex one is expensive because the number of vertices may be exponentially large. 
Both representations of a polytope have different utility, and hence it is important that the LP-based algorithm can deal with both. 

\paragraph*{MWU-based algorithm}
We improve the technique of \cite{krause2011,CLSS2017} to obtain an approximation algorithm based on the MWU method. 
Their algorithm adopts the value of $f_k(X) \ (k\in[n])$ for update, but this may lead the slow convergence when $f_k(X)$ is small for some $k$. 
To overcome the drawback, we make the convergence rate per iteration faster by introducing a novel concept called \emph{$(\eta,\gamma)$-reduction}. 
For any nonnegative function $f$, a function $g$ is called an $(\eta,\gamma)$-reduction of $f$ if
(i) $g(X)$ is always at most $\min\{f(X),\eta\}$ and
(ii) $f(X)=g(X)$ for any $X$ such that $g(X)$ is at most $\gamma\cdot\eta$. 
We assume that for some polynomially bounded $\gamma \leq 1$, there exists an $\alpha$-approximation algorithm that solves $\max_{X\in\cI}\sum_{k\in[n]}q_kf_k^\eta(X)$ for any $\eta$ and $q \in \Delta_n$, where $f_k^\eta$ is an $(\eta,\gamma)$-reduction of $f_k$ for each $k$. 
By using the approximation algorithm as a subroutine and by setting appropriately the value of $\eta$, we show that for any fixed constant $\epsilon > 0$, our scheme gives an $(\alpha-\epsilon)$-approximation solution in polynomial time with respect to $n$ and $1/\epsilon$. 
We remark that the support size of the output may be equal to the number of iterations. 
Without loss of the objective value, we can find a sparse solution whose support size is at most $n$ by using LP. 

The merit of the MWU-based algorithm is the applicability to a wide class of the robust optimization problem. 
We also demonstrate our scheme for various optimization problems. 
For any $\eta\ge 0$, we show that 
a linear function has an $(\eta,1/|E|)$-reduction to a linear function, 
a monotone submodular function has an $(\eta,1)$-reduction to a monotone submodular function, and 
a non-monotone submodular function has an $(\eta,1/|E|)$-reduction to a submodular function. 
Therefore, we can construct subroutines owing to existing work. 
Consequently, for the linear case, we obtain an FPTAS for \eqref{eq:problem} subject to the knapsack constraint and a $1/(\mu-1+\epsilon)$-approximation algorithm subject to the $\mu$-matroid intersection constraint. 
For the monotone submodular case, there exist a $(1-1/e-\epsilon)$-approximation algorithm for the knapsack or matroid constraint, and a $1/(\mu+\epsilon)$-approximation for the $\mu$-matroid intersection constraint. 
For the non-monotone submodular case, we derive a $(1/2-\epsilon)$-approximation algorithm for \eqref{eq:problem} without a constraint. 

An important application of our MWU-based scheme is the maximum cardinality robustness problem. 
For independence systems defined from the knapsack problem, we obtain an FPTAS for the maximum cardinality robustness problem. 
To construct the subroutine, we give a gap-preserving reduction of $\max_{X\in\cI}\sum_{k\in[n]}q_kf_k^\eta(X)$ to $\max_{X \in \cI} v_{\leq k}(X)$ for any $k$, which admits an FPTAS~\cite{CapraraKPP2000}. 
We also show that the maximum cardinality robustness problem is NP-hard. 

We remark that both schemes produce a randomized strategy, but the schemes themselves are deterministic.
Our results are summarized in Table~\ref{table:our_results}.

\begin{table}[htbp]
\centering
\caption{The approximation ratios for robust optimization problems shown in the present paper.}\label{table:our_results}
\resizebox{\textwidth}{!}{%
\begin{tabular}{l|llll}
\toprule
&    objective functions   &constraint &approximation ratio &reference\\\midrule
\parbox[t]{2mm}{\multirow{4}{*}{\rotatebox[origin=c]{90}{LP-based}}}
&    linear (polytope)     &matroid (intersection) &P &Thm.~\ref{thm:LP-based} \\
&    linear (polytope)     &shortest $s$--$t$ path &P &Thm.~\ref{thm:LP-based}\\
&    linear (polytope)     &knapsack   &PTAS&Thm.~\ref{thm:polytope_knapsack}\\
&    linear (polytope)     &$\mu$-matroid intersection &$\frac{2}{e\mu}$-approx. &Thm.~\ref{thm:polytope_mu}\\\hline
\parbox[t]{2mm}{\multirow{6}{*}{\rotatebox[origin=c]{90}{MWU-based}}} 
&    monotone submodular   &matroid/knapsack &$(1-\frac{1}{e}-\epsilon)$-approx.&Thm.~\ref{thm:monotone_submodular}\\
&    monotone submodular   &$\mu$-matroid intersection &$\frac{1}{\mu+\epsilon}$-approx.&Thm.~\ref{thm:monotone_submodular_mu}\\
&    linear                &$\mu$-matroid intersection &$\frac{1}{\mu-1+\epsilon}$-approx.&Thm.~\ref{thm:monotone_linear_mu}\\
&    submodular &free      &$(\frac{1}{2}-\epsilon)$-approx.&Thm.~\ref{thm:non-monotone_submodular}\\
&    linear                &knapsack   &FPTAS&Thm.~\ref{thm:monotone_linear_knapsack}\\
&    cardinality           &knapsack   &FPTAS&Thm.~\ref{thm:cardinality_knapsack}\\
\bottomrule
\end{tabular}}
\end{table}

 \paragraph*{Orgamization of this paper}
 The rest of this paper is organized as follows. 
 In Section~\ref{sec:preliminaries}, we fix notations and give a precise description of our problem.
 In Section~\ref{sec:lp-based}, we explain basic scheme of LP-based algorithms and then extend the result
 to a knapsack constraint case and a $\mu$-matroid intersection constraint case.
 In Section~\ref{sec:MWU}, we explain multiplicative weights update method.

\section{Preliminaries}\label{sec:preliminaries}

\paragraph*{Linear and submodular functions}
Throughout this paper, we consider set functions $f$ with $f(\emptyset)=0$. 
We say that a set function $f\colon 2^E\rightarrow \mathbb{R}$ is \emph{submodular} if $f(X)+f(Y)\geq f(X\cup Y) + f(X\cap Y)$ holds for all $X, Y \subseteq E$~\cite{fujishige2005,KG2014}. 
In particular, a set function $f\colon 2^E\rightarrow \mathbb{R}$ is called \emph{linear} (modular) if $f(X)+f(Y)= f(X\cup Y) + f(X\cap Y)$ holds for all $X, Y \subseteq E$. 
A linear function $f$ is represented as $f(X)=\sum_{e\in X}w_{e}$ for some $(w_e)_{e \in E}$. 
A function $f$ is said to be \emph{monotone} if $f(X) \leq f(Y)$ for all $X \subseteq Y \subseteq E$. 
A linear function $f(X)=\sum_{e\in X}w_{e}$ is monotone if and only if $w_e \geq 0$ ($e\in E$). 

\paragraph*{Independence system}
Let $E$ be a finite ground set. 
An \emph{independence system} is a set system $(E, \cI)$ with the following properties:
(I1) $\emptyset\in\mathcal{I}$, and 
(I2) $X\subseteq Y\in\mathcal{I}$ implies $X\in\mathcal{I}$. 
A set $I\subseteq \mathcal{I}$ is said to be \emph{independent}, and an inclusion-wise maximal independent set is called a \emph{base}.
The class of independence systems is wide and it includes matroids, $\mu$-matroid intersections, and families of knapsack solutions. 

A \emph{matroid} is an independence system $(E,\mathcal{I})$ satisfying that 
(I3) $X,Y\in\mathcal{I}$, $|X|<|Y|$ implies the existence of $e\in Y\setminus X$ such that $X\cup\{e\}\in\mathcal{I}$.
All bases of a matroid have the same cardinality, which is called the \emph{rank} of the matroid and is denoted by $\rho(\mathcal{I})$. 
An example of matroids is a uniform matroid $(E, \cI)$, where $\cI = \{ S \subseteq E \mid |S| \leq  r\}$ for some $r$. 
Note that the rank of this uniform matroid is $r$. 
Given two matroids $\mathcal{M}_1 = (E,\mathcal{I}_1)$ and $\mathcal{M}_2 = (E,\mathcal{I}_2)$, the \emph{matroid intersection} of $\mathcal{M}_1$ and $\mathcal{M}_2$ is defined by $(E,\mathcal{I}_1\cap\mathcal{I}_2)$. 
Similarly, given $\mu$ matroids $\mathcal{M}_i = (E,\mathcal{I}_i) \ (i=1, \ldots, \mu)$, the $\mu$-matroid intersection is defined by $(E,\bigcap_{i=1}^\mu \mathcal{I}_i)$. 

Given an item set $E$ with size $s(e)$ and value $v(e)$ for each $e \in E$, and the capacity $C \in \mathbb{Z}_+$, the \emph{knapsack problem} is to find a subset $X$ of $E$ that maximizes the total value $\sum_{e \in X} v(e)$ subject to a knapsack constraint $\sum_{e \in X} s(e) \leq C$. 
Each subset satisfying the knapsack constraint is called a knapsack solution. 
Let $\cI = \{ X \mid \sum_{e \in X} s(e) \leq C \}$ be the family of knapsack solutions. 
Then, $(E, \cI)$ is an independence system. 

\paragraph*{Robust optimization problem}
Let $E$ be a finite ground set, and let $n$ be a positive integer. 
Given $n$ set functions $f_1,\dots,f_n\colon 2^E\to\mathbb{R}_+$
and an independence system $(E,\cI)$,
our task is to solve
\[\max_{p\in\Delta(\cI)}\min_{k\in[n]}\sum_{X\in\cI}p_X\cdot f_k(X).\]
For each $k \in [n]$, we denote $X_k^*\in\argmax_{X\in\cI}f_k(X)$ and assume that $f_k(X_k^*)>0$. 
We assume that the functions are given by an \emph{oracle}, i.e.,
for a given $X\subseteq E$, we can query an oracle about the values $f_1(X),\dots,f_n(X)$.
Let $\Delta(\cI)$ and $\Delta_n$ denote the set of probability distributions over $\cI$ and $[n]$, respectively.

By von Neumann's minimax theorem~\cite{vonneumann1928}, it holds that
\begin{align}
  \max_{p\in\Delta(\cI)}\min_{k\in[n]}\sum_{X\in\cI} p_X\cdot f_k(X)
  =\min_{q\in \Delta_n}\max_{X\in\cI}\sum_{k\in[n]} q_k\cdot f_k(X). \label{eq:minimax}
\end{align}
This leads the following proposition, which will be used later. 
\begin{proposition}\label{prop:nustar}
Let $\nu^*$ denote the optimal value of \eqref{eq:problem}. 
It holds that $\min_{k\in[n]}f_k(X_k^*)/n\le \nu^*\le \min_{k\in[n]}f_k(X_k^*)$.
\end{proposition}
\begin{proof}
	The upper bound follows from
	\begin{align*}
	\nu^*
	&=\min_{q\in \Delta_n}\max_{X\in\cI}\sum_{k\in[n]} q_k\cdot f_k(X)
	\le \min_{k\in [n]}\max_{X\in\cI} f_k(X)
	=\min_{k\in [n]}f_k(X_k^*).
	\end{align*}
	
	Let $p^*\in\Delta(\cI)$ be a probability distribution such that $p^*_X=|\{i\in[n]\mid X_i^*=X\}|/n$.
	Then we have
	\begin{align*}
	\nu^*
	&=\max_{p\in\Delta\cI}\min_{k\in[n]}\sum_{X\in\cI} p_X\cdot f_k(X)
	\ge \min_{k\in[n]}\sum_{X\in\cI} p^*_X\cdot f_k(X)
	\ge \min_{k\in[n]} f_k(X_k^*)/n.
	\end{align*}
\end{proof}
This implies that we can find a $1/n$-approximate solution by just computing $X^*_k \ (k \in [n])$. 

We prove that, even for an easy case, computing the optimal worst case value among deterministic solutions is strongly NP-hard even to approximate.
To prove this, we reduce the \emph{hitting set problem}, which is known to be NP-hard~\cite{garey1979cai}.
Given $n$ subsets $S_k\subseteq E$ $(k\in[n])$ on a ground set $E$ and an integer $r$, the hitting set problem is to find a subset $A\subseteq E$ such that $|A|\le r$ and $S_k\cap A\ne\emptyset$ for all $k\in[n]$.
\begin{theorem}\label{thm:dethard}
	It is NP-hard to compute 
	\begin{align}\label{eq:deterministic}
	\max_{X\in\cI}\min_{k\in[n]}f_k(X)
	\end{align} even
when the objective functions $f_1,\dots,f_k$ are linear and $\cI$ is given by a uniform matroid.
Moreover, there exists no approximation algorithm for the problem unless P=NP.
\end{theorem}
\begin{proof}
	Let $(E, \{S_1, \ldots, S_n\}, r)$ be an instance of the hitting set problem. 
	We construct an instance of \eqref{eq:deterministic} as follows. 
	The constraint $\cI$ is defined so that $(E, \cI)$ is the rank $r$ uniform matroid. 
	Note that $\cI$ is a family of subsets with at most $r$ elements. 
	Each objective function $f_k$ $(k\in[n])$ is defined by $f_k(X)=|X\cap S_k|$ $(X\subseteq E)$, which is linear.
	
	If there exists an hitting set $X\in\cI$, then $\min_{k\in[n]} f_k(X)\ge 1$, which implies that the optimal value of \eqref{eq:deterministic} is at least $1$. 
	On the other hand, if any $X\in\cI$ is not a hitting set, then $\min_{k\in[n]} f_k(X)=0$ for all $X\in\cI$, meaning that the optimal value of \eqref{eq:deterministic} is $0$. 
	Therefore, even deciding whether the optimal value of \eqref{eq:deterministic} is positive or zero is NP-hard. 
	Thus, there exists no approximation algorithm to the problem unless P=NP.
\end{proof}

\section{LP-based Algorithms}\label{sec:lp-based}
In this section, we present a computation scheme for the robust optimization problem \eqref{eq:problem}
with linear functions $f_1,\dots,f_n$, i.e., $f_k(X)=\sum_{e\in X}w_{ke}$.
Here, $w_{ke}\ge 0$ holds for $k\in[n]$ and $e\in E$ since we assume $f_k(X)\ge 0$.
A key technique is the separation problem for an independence system polytope. 
An \emph{independence system polytope} of $(E, \cI)$ is a polytope defined as $P(\cI)=\conv\{\chi(X)\mid X\in\cI\}\subseteq[0,1]^{E}$, where $\chi(X)$ is a characteristic vector in $\{0,1\}^E$, i.e., $\chi(X)_e=1$ if and only if $e\in X$. 
For a probability distribution $p\in\Delta(\cI)$, we can get a point $x\in P(\cI)$ such that $x_e=\sum_{X\in\cI:\,e\in X}p_X$ $(e\in E)$.
Then, $x_e \ (e \in E)$ means a probability that $e$ is chosen when we select an independent set according to the probability distribution $p$.
Conversely, for a point $x \in P(\cI)$, there exists $p\in\Delta(\cI)$ such that $\sum_{X \in \cI} p_X \, \chi(X) = x$ by definition of $P(\cI)$. 
Given $x\in\mathbb{R}^E$, the separation problem for $P(\cI)$ is to either assert $x\in P(\cI)$ or
find a vector $d$ such that $d^\top x<d^\top y$ for all $y\in P(\cI)$.

The rest of this section is organized as follows. 
In Section~\ref{sec:LP-based basic}, we prove that we can solve \eqref{eq:problem} in polynomial time
if there is a polynomial-time algorithm to solve the separation problem for $P(\cI)$. 
We list up classes of independence systems such that there exists a polynomial- time algorithm for the separation problem in Section~\ref{subsec:isp}. 
In Section~\ref{sec:LP-based relax}, we tackle the case when it is hard to construct a separation algorithm for $P(\cI)$. 
We show that we can obtain an approximation solution when we can slightly relax $P(\cI)$.
Moreover, we deal with a setting that objective functions are given by a polytope in Section~\ref{sec:LP-based polytope}, and consider nearly linear functions $f_1, \ldots, f_n$ in Section~\ref{sec:LP-based approx linear}. 

\subsection{Basic scheme}\label{sec:LP-based basic}
We observe that the optimal robust value of \eqref{eq:problem} is the same as the optimal value of
the following linear programming (LP):
\begin{align}\label{eq:lp}
\max~\nu\quad\text{s.t.}\quad\nu\le \sum_{e\in E}w_{ie} x_e \quad (\forall i\in[n]),\quad x\in P(\cI).
\end{align}

\begin{lemma}\label{lem:LP}
When $f_1, \ldots, f_n$ are linear, the optimal value of \eqref{eq:lp} is equal to that of \eqref{eq:problem}. 
\end{lemma}
\begin{proof}
	Let $p^*\in\Delta(\cI)$ be the optimal solution of \eqref{eq:problem} and let $\nu^*=\min_{k\in[n]}\sum_{X\in\cI}p_X^*\cdot f_k(X)$. 
	Let $x^*\in\mathbb{R}^E$ be a vector such that $x^*_e=\sum_{X\in\cI:\,e\in X}p^*_X$ for each $e\in E$. 
	Note that $x^* = \sum_{X\in\cI}p^*_X \chi(X)$. 
	Then $x^*\in P(\cI)$ holds by the definition of $P(\cI)$.
	Thus, the optimal value of \eqref{eq:lp} is at least
	\begin{align*}
	\min_{k\in[n]}\sum_{e\in E}w_{ke} x_e^*
	&=\min_{k\in[n]}\sum_{e\in E}\sum_{X\in\cI:\,e\in X}p_X^*\cdot w_{ke}\\
	&=\min_{k\in[n]}\sum_{X\in\cI}\sum_{e\in X}p_X^*\cdot w_{ke}
	=\min_{k\in[n]}\sum_{X\in\cI}p_X^*\cdot f_k(X)=\nu^*.
	\end{align*}
	
	On the other hand, let $(\nu',x')$ be an optimal solution of \eqref{eq:lp}.
	As $x'\in P(\cI)~(=\conv\{\chi(X)\mid X\in\cI\})$,
	there exists a $p'\in\Delta(\cI)$ such that $x'_e=\sum_{X\in\cI:\,e\in X}p'_X$ for each $e\in E$.
	Then we have
	\begin{align*}
	\nu^*
	&=\max_{p\in\Delta(\cI)}\min_{k\in[n]}\sum_{X\in\cI}p_X\cdot f_k(X)\\
	&\ge \min_{k\in[n]}\sum_{X\in\cI}p_X'\cdot f_k(X)
	=\min_{k\in[n]}\sum_{X\in\cI}\sum_{e\in X}p_X'\cdot w_{ke}\\
	&=\min_{k\in[n]}\sum_{e\in E}w_{ke}\sum_{X\in\cI:e\in X}p_X'
	=\min_{k\in[n]}\sum_{e\in E}w_{ke}\cdot x_e'
	\ge\nu'.
	\end{align*}
\end{proof}

Thus the optimal solution of \eqref{eq:problem} is obtained by the following two-step scheme. 
\begin{enumerate}
\item compute the optimal solution of LP~\eqref{eq:lp}, which we denote as $(\nu^*,x^*)$,
\item compute $p^*\in\Delta(\cI)$ such that $x^* = \sum_{X\in\cI} p^*_X \, \chi(X)$. 
\end{enumerate}

It is trivial that if $|\cI|$ is bounded by a polynomial in $|E|$ and $n$, 
then we can obtain $p^*$ by replacing $x$ with $\sum_{X\in\cI} p_X \chi(X)$ in \eqref{eq:lp} and solving it. 
In general, we can solve the two problems in polynomial time by the ellipsoid method
when we have a polynomial-time algorithm to solve the separation problem for $P(\cI)$.
This is due to the following theorems given by Gr{\"o}tschel, Lov{\'a}sz, and Schrijver~\cite{GLS2012}.
\begin{theorem}[\cite{GLS2012}]\label{thm:GLS1}
  Let $\cP\subseteq\mathbb{R}^E$ be a polytope.
  If the separation problem for $\cP$ can be solved in polynomial time,
  then we can solve a linear program over $\cP$ in polynomial time.
\end{theorem}
\begin{theorem}[\cite{GLS2012}]\label{thm:GLS2}
  Let $\cP\subseteq\mathbb{R}^E$ be a polytope.
  If the separation problem for $\cP$ can be solved in polynomial time,
  then there exists a polynomial time algorithm that, for any vector $x\in\cP$,
  computes affinely independent vertices $x_1,\dots,x_{\ell}$ of $\cP$ ($\ell\le |E|+1$)
  and positive reals $\lambda_1,\dots,\lambda_{\ell}$ with $\sum_{i=1}^\ell\lambda_i=1$ such that
  $x=\sum_{i=1}^{\ell}\lambda_i x_i$.
\end{theorem}

Therefore, we see the following general result.
\begin{theorem}\label{thm:LP-based}
  If $f_1,\dots,f_n$ are linear and there is a polynomial-time algorithm to solve the separation problem for $P(\cI)$,
  then we can solve the linear robust optimization problem \eqref{eq:problem} in polynomial time.
\end{theorem}

\subsection{Independence system polytopes with separation algorithms}\label{subsec:isp}
Here we list up classes of independence systems such that there exists a polynomial-time algorithm for the separation problem.
For more details of the following representation of polytopes, see \cite{schrijver2003}.

\paragraph*{Matroid constraint}
Suppose that $(E,\cI)$ is a matroid with a rank function $\rho$.
Then, we can write
\begin{align*}
P(\cI)=\left\{x\in [0,1]^E\,\middle|\, \sum_{e\in U}x_e\le \rho(U)~(\forall U\subseteq E)\right\}.
\end{align*}
The separation problem for $P(\cI)$ is solvable in strongly polynomial time by Cunningham's algorithm~\cite{cunningham1984}.
Thus, we can solve the linear robust optimization problem \eqref{eq:problem} subject to a matroid constraint.

\paragraph*{Matroid intersection}
Let $(E,\cI_1)$ and $(E,\cI_2)$ be matroids with rank functions $\rho_1$ and $\rho_2$.
Then, for a matroid intersection $(E,\cI)=(E,\cI_1\cap \cI_2)$,
we can denote
\begin{align*}
P(\cI)=\left\{x\in [0,1]^E\,\middle|\, \sum_{e\in U}x_e\le \rho_i(U)~(\forall U\subseteq E,~i=1,2)\right\}
\end{align*}
and hence the separation problem for $P(\cI)$ is solvable in strongly polynomial time by Cunningham's algorithm~\cite{cunningham1984}.
Thus, we can solve the linear robust optimization problem \eqref{eq:problem} subject to a matroid intersection constraint.
We remark that matroid intersection includes bipartite matching and arborescences in directed graphs
and hence we can also solve the robust maximum weight bipartite matching problem
and the robust maximum weight arborescence problem.

\paragraph*{Shortest $s$--$t$ path}
We explain that our scheme works for the set of $s$--$t$ paths, although it does not form a independence system. 
We are given a directed graph $G=(V,E)$, source $s\in V$, destination $t\in V$, and length $\ell_k\colon E\to\mathbb{R}_{++}$ $(k\in[n])$.
Let $\cI\subseteq 2^E$ be the set of $s$--$t$ paths and $f_k(X)=\sum_{e\in X}\ell_k(e)$ for $k\in[n]$.
Then, our task is to find a probability distribution over $s$--$t$ paths $p\in \Delta(\cI)$ that minimizes $\max_{k\in[n]}\sum_{X\in\cI}p_X f_k(X)$.
We mention that the deterministic version of this problem is NP-hard even for restricted cases~\cite{YY1998}. 

Since the longest $s$--$t$ path problem is NP-hard, we cannot expect an efficient separation algorithm for $P(\cI)$.
However, if we extend the $s$--$t$ path polytope to its dominant it becomes tractable.
The dominant $P^\uparrow(\cI)$ of $P(\cI)$ is defined as the set of vectors $x\in\mathbb{R}^E$ with $x\ge y$ for some $y\in P(\cI)$.
Then, we can denote
\begin{align*}
P^\uparrow(\cI)=\left\{x\in \mathbb{R}_+^E\,\middle|\, \sum_{e\in U}x_e\ge 1~(\forall U\subseteq E:~\text{$U$ is an $s$--$t$ cut})\right\}.
\end{align*}
The separation problem for the polytope $P^\uparrow(\cI)\cap [0,1]^E$ can be solved in polynomial time by solving a minimum $s$--$t$ cut problem
and hence we can obtain
\[x^*\in\argmin_{x\in P^\uparrow(\cI)\cap[0,1]^E}\max_{k\in[n]}\sum_{e\in E}\ell_k(e)x_e.\]
Moreover, since $\ell_k(e)>0$ for all $k\in[n]$ and $e\in E$,
we have $x^*\in P(\cI)$, and hence we can obtain the optimal solution of the robust shortest path problem $\min_{p\in\Delta(\cI)}\max_{k\in[n]}f_k(X)$.

\subsection{Relaxation of the polytope}\label{sec:LP-based relax}
We present an approximation scheme for the case when the separation problem for $P(\cI)$ is hard to solve. 
Recall that $f_k(X)=\sum_{e\in X}w_{ke}$ where $w_{ke}\ge 0$ for $k\in[n]$ and $e\in E$. 

We modify the basic scheme as follows. 
First, instead of solving the separation problem for $P(\cI)$, we solve the one for a relaxation of $P(\cI)$. 
For a polytope $P$ and a positive number $(1\ge )\,\alpha > 0$, we denote $\alpha P = \{\alpha x \mid x \in P \}$. 
We call a polytope $\hat{P}(\cI)\subseteq [0,1]^E$ $\alpha$-relaxation of $P(\cI)$ if it holds that 
\begin{align*}
\alpha \hat{P}(\cI)\subseteq P(\cI)\subseteq \hat{P}(\cI). 
\end{align*}
Then we solve 
\begin{align}\label{eq:lp relax}
\max_{x\in \hat{P}(\cI)}\min_{k\in[n]}\sum_{e\in E}w_{ke}x_e
\end{align}
instead of LP \eqref{eq:lp}, and obtain an optimal solution $\hat{x}$. 

Next, we compute a convex combination of $\hat{x}$ using $\chi(X) \ (X \in \cI)$. 
Here, if $\hat{x}\in\hat{P}(\cI)$ is the optimal solution for \eqref{eq:lp relax},
then $\alpha\hat{x}\in P(\cI)$ is an $\alpha$-approximate solution of LP \eqref{eq:lp}, because 
\begin{align*}
  \max_{x\in P(\cI)}\min_{k\in[n]}\sum_{e\in E}w_{ke}x_e
  &\le \max_{x\in \hat{P}(\cI)}\min_{k\in[n]}\sum_{e\in E}w_{ke}x_e \\
  &=\min_{k\in[n]}\sum_{e\in E}w_{ke}\hat{x}_e
  =\frac{1}{\alpha}\cdot\min_{k\in[n]}\sum_{e\in E}w_{ke}(\alpha\hat{x}_e).
\end{align*}
As $\alpha\hat{x}\in P(\cI)$, there exists $p\in\Delta(\cI)$ such that $\alpha\hat{x}=\sum_{X\in\cI} p_X \, \chi(X)$.
However, the retrieval of such a probability distribution may be computationally hard, because the separation problem for $P(\cI)$ is hard to solve. 
Hence, we relax the problem and compute $p^*\in\Delta(\cI)$ such that $\beta\hat{x}\le\sum_{X\in\cI} p^*_X \, \chi(X)$, where $(\alpha\ge)\beta>0$.
Then, $p^*$ is a $\beta$-approximate solution of $\max_{p\in\Delta(\cI)}\min_{k\in[n]}\sum_{X\in\cI}p_X^*\cdot f_k(X)$, 
because
\begin{align*}
  \max_{p\in\Delta(\cI)}\min_{k\in[n]}\sum_{X\in\cI}p_X\cdot f_k(X)
  &\le\min_{k\in[n]}\sum_{e\in E}w_{ke}\hat{x}_e 
  =\frac{1}{\beta}\cdot\min_{k\in[n]}\sum_{e\in E}w_{ke}(\beta\hat{x}_e)\\
  &\le \frac{1}{\beta}\cdot\min_{k\in[n]}\sum_{e\in E}\sum_{X\in\cI:\,e\in X}p^*_Xw_{ke} \\
  &=\frac{1}{\beta}\cdot\min_{k\in[n]}\sum_{X\in\cI}p^*_X\cdot f_k(X).
\end{align*}

Thus the basic scheme is modified as the following approximation scheme:
\begin{enumerate}
\item compute the optimal solution $\hat{x}\in\hat{P}(\cI)$ for LP \eqref{eq:lp relax},
\item compute $p^*\in\Delta(\cI)$ such that $\beta\cdot\hat{x}_e\le \sum_{X\in\cI:\,e\in X}p^*_X$ for each $e\in E$. 
\end{enumerate}

\begin{theorem}\label{thm:lp relax}
Suppose that $f_1, \ldots, f_n$ are linear. 
If there exists a polynomial-time algorithm to solve the separation problem for a $\alpha$-relaxation $\hat{P}(\cI)$ of $P(\cI)$, then an $\alpha$-approximation of the optimal value of \eqref{eq:problem} is computed in polynomial-time. 
In addition, if there exists a polynomial-time algorithm to find $p\in\Delta(\cI)$ such that $\beta\cdot\hat{x}_e\le \sum_{X\in\cI:\,e\in X}p_X$ for any $x \in \hat{P}(\cI)$, then a $\beta$-approximate solution of \eqref{eq:problem} is found in polynomial-time. 
\end{theorem}

We remark that we can combine the result in Section~\ref{sec:LP-based polytope} with this theorem.

In the subsequent sections, we apply Theorem \ref{thm:lp relax} to important two cases when $\cI$ is defined from a knapsack constraint or a $\mu$-matroid intersection.
For this purpose, we develop appropriate relaxations of $P(\cI)$ and retrieval procedures for $p^*$.

\subsubsection{Relaxation of a knapsack polytope}
Let $E$ be a set of items with size $s(e)$ for each $e \in E$. 
Without loss of generality, we assume that a knapsack capacity is one, and $s(e)\le 1$ for all $e\in E$.
Let $\cI$ be a family of knapsack solutions, i.e., $\cI=\{X\subseteq E\mid \sum_{e\in X}s(e)\le 1\}$. 

It is known that $P(\cI)$ admits a \emph{polynomial size relaxation scheme (PSRS)}, i.e., there exists a $(1-\epsilon)$-relaxation of $P(\cI)$ through a linear program of polynomial size for a fixed $\epsilon>0$. 
\begin{theorem}[Bienstock \cite{bienstock2008}]
  Let $0<\epsilon\le 1$.
  There exist a polytope $P^\epsilon(\cI)$ and its extended formulation 
  with $O(\epsilon^{-1}n^{1+\lceil 1/\epsilon\rceil})$ variables and
  $O(\epsilon^{-1}n^{2+\lceil 1/\epsilon\rceil})$ constraints such that 
  \begin{align*}
    (1-\epsilon)P^\epsilon(\cI)\subseteq P(\cI)\subseteq P^\epsilon(\cI).
  \end{align*}
\end{theorem}
Thus, the optimal solution $\hat{x}$ to $\max_{x\in P^\epsilon(\cI)}\min_{k\in[n]}\sum_{e\in E}w_{ke}x_e$ can be computed in polynomial time. 
The remaining task is to compute $p^*\in\Delta(\cI)$ such that $(1-\epsilon)\cdot\hat{x}_e\le \sum_{X\in\cI:\,e\in X}p^*_X$ for each $e\in E$.
We give an algorithm for this task. 
\begin{lemma}\label{lem:knapsack_decompose}
There exists a polynomial-time algorithm that computes $p^*\in\Delta(\cI)$ such that $(1-\epsilon)\cdot\hat{x}_e\le \sum_{X\in\cI:\,e\in X}p^*_X$ for each $e\in E$.
\end{lemma}
\begin{proof}
	To obtain such a probability distribution, we explain Bienstock's relaxation scheme.
	Let $\kappa=\lceil 1/\epsilon\rceil$ and 
	let $\cS_i=\{S\subseteq E\mid |S|=i,~\sum_{e\in S}s(e)\le 1\}$ for $i=1,\dots,\kappa$.
	Then, the constraints of $P^\epsilon(\cI)$ are given as follows:
	\begin{align}
	&\textstyle x_e=\sum_{i=1}^\kappa\sum_{S\in \cS_i}y_e^S&&(\forall e\in E),\label{eq:psrs1}\\
	&y_e^S=y_0^S &&(\forall S\in \textstyle\bigcup_{i=1}^\kappa\cS_i,\ \forall e\in S),\label{eq:psrs2}\\
	&y_e^S=0     &&(\forall S\in \textstyle\bigcup_{i=1}^{\kappa-1}\cS_i,\ \forall e\in E\setminus S),\label{eq:psrs3}\\
	&y_e^S\le y_0^S  &&(\forall S\in \textstyle\bigcup_{i=1}^{\kappa}\cS_i,\ \forall e\in E\setminus S),\label{eq:psrs4}\\
	&\textstyle\sum_{e\in E}s(e)y_e^S\le y_0^S &&(\forall S\in \cS_\kappa),\label{eq:psrs5}\\
	&y_e^S=0         &&(\forall S\in \cS_\kappa,\ \forall e\in E\setminus S: s(e)>\min_{e'\in S}s(e')),\label{eq:psrs6}\\
	&y_e^S\ge 0  &&(\forall S\in \textstyle\bigcup_{i=1}^\kappa\cS_i,\ \forall e\in E\cup\{0\}),\label{eq:psrs7}\\
	&\textstyle \sum_{i=1}^\kappa\sum_{S\in \cS_i}y_0^S=1.&&\label{eq:psrs8}
	\end{align}
	Intuitively, $y^S$ corresponds to a knapsack solution $S\in\cI$ if $|S|<\kappa$
	and $y^S$ corresponds to a (fractional) knapsack solution such that $S$ is the $\kappa$-largest items if $|S|=\kappa$.

	Let $\hat{x}$ be an optimal solution for $\max_{x\in P^\epsilon(\cI)}\min_{k\in[n]}\sum_{e\in E}w_{ke}x_e$
	and let $(\hat{x},\hat{y})$ satisfy \eqref{eq:psrs1}--\eqref{eq:psrs8}.
	For each $S\in\cS_\kappa$, we define
	\begin{align*}
	Q^S=\left\{\ y\ \middle|\ 
	\begin{array}{l}
	\sum_{e\in E}s(e)y_e\le 1,\\
	y_e=1 \qquad(\forall e\in S),\\
	y_e=0 \qquad(\forall e\in E\setminus S: s(e)>\min_{e'\in S}s(e')),\\
	0\le y_e\le 1\quad(\forall e\in E\setminus S: s(e)\le \min_{e'\in S}s(e'))
	\end{array}\right\}.
	\end{align*}
	Let us denote $\hat{y}^S=(y^S_e)_{e\in E}$.
	Then, by \eqref{eq:psrs2} and \eqref{eq:psrs4}--\eqref{eq:psrs7}, we have $\hat{y}^S\in \hat{y}_0^S Q^S$.
	Also, by Theorem~\ref{thm:GLS2}, 
	we can compute a convex combination representation of $\hat{y}^S$
	with at most $|E|-\kappa+1~(\le |E|)$ vertices of $\hat{y}^S_0Q^S$
	for each $S\in\cS_\kappa$ with $\hat{y}^S_0>0$.
	Suppose that $\hat{y}^S=\hat{y}^S_0\cdot\sum_{i=1}^{t_S}\lambda^{S,i}\tilde{y}^{S,i}$
	where $t_S\le |E|$,
	$\tilde{y}^{S,i}$ is a vertex of $Q^S$,
	$\sum_{i=1}^{t_S}\lambda^{S,i}=\hat{y}^S_0$,
	and $\lambda^{S,i}\ge 0$ ($i=1,\dots,t_S$).
	
	Let $\tilde{y}$ be a vertex of $Q^S$ that is not integral.
	Then, there exists exactly one item $e^*$ such that $0<\tilde{y}_{e^*}<1$~\cite{KPP2004}.
	Let $T=\{e\in E\mid \tilde{y}_e>0\}$.
	Then, $T\setminus\{e\}\in\cI$ for every $e\in S\cup\{e^*\}$ and it holds that
	\begin{align*}
	\frac{\kappa}{\kappa+1}\tilde{y}\le \sum_{e\in S\cup\{e^*\}}\frac{1}{\kappa+1}\chi(T\setminus\{e\})
	\end{align*}
	because, if $\tilde{y}_{t}>0$ (i.e., $t\in T$),
	we have \(\sum_{e\in S\cup\{e^*\}}\chi(T\setminus\{e\})_t=|S\cup\{e^*\}\setminus\{t\}|\ge\kappa\).

	Now, we are ready to construct the probability distribution $p^*$.
	Let us define
	\begin{align*}
	p^*=&\sum_{i=1}^{\kappa-1}\sum_{S\in\cS_i}\hat{y}^S_0\chi(S)
	+\sum_{S\in\cS_\kappa:\,\hat{y}^S_0>0}\hat{y}^S_0\sum_{i\in[t_S]:\,\tilde{y}^{S,i}\text{ is integral}}\lambda^{S,i} \chi(\supp(\tilde{y}^{S,i}))\\
	&+\sum_{S\in\cS_\kappa:\,\hat{y}^S_0>0}\hat{y}^S_0\sum_{i\in[t_S]:\,0<\tilde{y}^{S,i}_{e^*}<1}\frac{\lambda^{S,i}}{\kappa+1}\sum_{e\in S\cup\{e^*\}}\chi(\supp(\tilde{y}^{S,i})\setminus\{e\}).
	\end{align*}
	Note that $p^*\in\Delta(\cI)$ because the sum of the coefficients is
	\begin{align*}
	\sum_{i=1}^{\kappa-1}&\sum_{S\in\cS_i}\hat{y}^S_0
	+\sum_{\substack{S\in\cS_\kappa:\\ \hat{y}^S_0>0}}\hat{y}^S_0\sum_{\substack{i\in[t_S]:\\\tilde{y}^{S,i}\text{ is integral}}}\lambda^{S,i}
	+\sum_{\substack{S\in\cS_\kappa:\\\hat{y}^S_0>0}}\hat{y}^S_0\sum_{\substack{i\in[t_S]:\\0<\tilde{y}^{S,i}_{e^*}<1}}\frac{\lambda^{S,i}}{\kappa+1}\sum_{e\in S\cup\{e^*\}}1=1.
	\end{align*}
	By \eqref{eq:psrs1}, we have
	\begin{align*}
	\hat{x}
	=\sum_{i=1}^\kappa\sum_{S\in \cS_i}\hat{y}^S
	&=\sum_{i=1}^{\kappa-1}\sum_{S\in \cS_i}\hat{y}^S
	+\sum_{S\in \cS_\kappa:\,\hat{y}_0^S>0}\hat{y}^S\\
	&=\sum_{i=1}^{\kappa-1}\sum_{S\in \cS_i}\hat{y}_0^S\chi(S)
	+\sum_{S\in \cS_\kappa:\,\hat{y}_0^S>0}\hat{y}_0^S\sum_{i=1}^{t_S}\lambda^{S,i}\tilde{y}^{S,i},
	\end{align*}
	which implies $\frac{\kappa}{\kappa+1}\cdot \hat{x}_e\le \sum_{X\in\cI:\,e\in X}p^*_X$.
	Here, $\kappa/(\kappa+1)\ge 1/(1+\epsilon)>1-\epsilon$,
	and hence $p^*$ is a $(1-\epsilon)$-approximate solution.
\end{proof}

\begin{theorem}\label{thm:polytope_knapsack}
There is a PTAS to compute the linear robust optimization problem \eqref{eq:problem} subject to a knapsack constraint.
\end{theorem}

Finally, we remark that the existence of a \emph{fully polynomial size relaxation scheme (FPSRS)} for $P(\cI)$ is open~\cite{bienstock2008}.
The existence of an FPSRS leads an FPTAS
to compute the optimal value of the linear robust optimization problem \eqref{eq:problem} subject to a knapsack constraint.

\subsubsection{Relaxation of a $\mu$-matroid intersection polytope}
Let us consider the case where $\cI$ is defined from a $\mu$-matroid intersection.
It is NP-hard to maximize a linear function subject to a $\mu$-matroid intersection constraint if $\mu\ge 3$~\cite{garey1979cai}.
Hence, it is also NP-hard to solve the linear robust optimization subject to a $\mu$-matroid intersection constraint if $\mu\ge 3$.
For $i=1, \ldots, \mu$, let $(E,\cI_i)$ be a matroid whose rank function is $\rho_i$. 
Let $(E,\cI)=(E,\bigcap_{i\in[\mu]} \cI_i)$. 
We define $\hat{P}(\cI)=\bigcap_{i\in[\mu]}P(\cI_i)$, i.e.,
\begin{align*}
  \hat{P}(\cI)=\left\{
  x\,\middle|
  \begin{array}{ll}
    \sum_{e\in X}x_e\le \rho_i(X)&(\forall i\in[\mu],~\forall X\subseteq E),\\
    x_e\ge 0&(\forall e\in E)
  \end{array}
  \right\}.
\end{align*}
Note that $P(\cI)=P(\bigcap_{i\in[\mu]}\cI_i)\subseteq \bigcap_{i\in[\mu]}P(\cI_i)=\hat{P}(\cI)$.
We see that $\hat{P}(\cI)$ is a $(1/\mu)$-relaxation of $P(\cI)$. 
\begin{lemma}\label{lem:mu-approx-polytope}
\(\frac{1}{\mu}\hat{P}(\cI)\subseteq P(\cI) \subseteq\hat{P}(\cI).\)
\end{lemma}
\begin{proof}
	To see this, we consider the following greedy algorithm for a given nonnegative weights $w$:
	start from the empty solution and process the elements in decreasing weight order,
	add an element to the current solution if and only if its addition preserves independence.
	It is known that the greedy algorithm is a $(1/\mu)$-approximation algorithm
	even with respect to the LP relaxation $\max_{x\in\hat{P}(\cI)}\sum_{e\in E}w(e)x_e$~\cite{FNW1978}.
	More precisely, for a weight $w\colon E\to\mathbb{R}_+$, we have
	\begin{align}
	\frac{1}{\mu}\cdot\max_{x\in\hat{P}(\cI)}\sum_{e\in E}w(e)x_e
	\le \sum_{e\in X^w}w(e)
	\le \max_{x\in P(\cI)}\sum_{e\in E}w(e)x_e
	\label{eq:FNW1978}
	\end{align}
	where $X^w$ is a greedy solution for $w$.
	It is sufficient to claim that $x\in\hat{P}(\cI)$ implies $\frac{1}{\mu}\cdot x\in P(\cI)$.
	To obtain a contradiction, suppose that $\hat{x}\in\hat{P}(\cI)$ but $\frac{1}{\mu}\cdot \hat{x}\not\in P(\cI)$.
	Then, by the separating hyperplane theorem~\cite{BV2004}, there exists a weight $\hat{w}\colon E\to\mathbb{R}$ such that
	\begin{align*}
	\max_{x\in P(\cI)}\sum_{e\in E}\hat{w}(e)x_e<\frac{1}{\mu}\sum_{e\in E}\hat{w}(e)\hat{x}_e\le \frac{1}{\mu}\max_{x\in\hat{P}(\cI)}\sum_{e\in E}\hat{w}(e)x_e.
	\end{align*}
	Let $\hat{w}^+\colon E\to\mathbb{R}_+$ be a weight such that $\hat{w}^+(e)=\max\{0,\hat{w}(e)\}$.
	Then, we see that $\max_{x\in P(\cI)}\sum_{e\in E}\hat{w}(e)x_e=\max_{x\in P(\cI)}\sum_{e\in E}\hat{w}^+(e)x_e$ because $\cI$ is downward closed.
	Also, we have $\sum_{e\in E}\hat{w}(e)\hat{x}_e\le \max_{x\in\hat{P}(\cI)}\sum_{e\in E}\hat{w}^+(e)\hat{x}_e$
	because $\hat{x}\ge 0$.
	Thus, we obtain
	\begin{align*}
	\max_{x\in P(\cI)}\sum_{e\in E}\hat{w}^+(e)x_e<\frac{1}{\mu}\sum_{e\in E}\hat{w}^+(e)\hat{x}_e
	\le \frac{1}{\mu}\max_{x\in\hat{P}(\cI)}\sum_{e\in E}\hat{w}^+(e)x_e
	\end{align*}
	which contradicts \eqref{eq:FNW1978}.
\end{proof}

As we can solve the separation problem for $\hat{P}(\cI)$ in strongly polynomial time~\cite{cunningham1984},
we can obtain an optimal solution $\hat{x}\in\hat{P}(\cI)$ for the relaxed problem $\max_{x\in \hat{P}(\cI)}\sum_{e\in E}w(e)x_e$.
Since $\hat{x}/\mu\in P(\cI)$, the value $\sum_{e\in E}w(e)\hat{x}_e/\mu$ is a $\mu$-approximation of the optimal value of \eqref{eq:problem}.

To obtain a $\mu$-approximate solution,
we need to compute $p^*\in\Delta(\cI)$
such that $\hat{x}_e/\mu \le \sum_{X\in\cI:\,e\in X}p^*_X$ for each $e\in E$.
Unfortunately, it seems hard to calculate such a distribution. 
With the aid of the contention resolution (CR) scheme, which is introduced by Chekuri, Vondr{\'a}k, and Zenklusen~\cite{CVZ2014}, we can compute $p^*\in\Delta(\cI)$ such that $(2/e\mu) \cdot \hat{x}_e \le \sum_{X\in\cI:\,e\in X}p^*_X$ for each $e\in E$. 
We describe its procedure later. 
We can summarize our result as follows. 
\begin{theorem}\label{thm:polytope_mu}
  We can compute a $\mu$-approximate value of the linear robust optimization problem subject to a $\mu$-matroid intersection in polynomial time.
  Moreover, we can implement a procedure that efficiently outputs an independent set according to the distribution of a $2/(e\mu)$-approximate solution.
\end{theorem}

It remains to describe the procedure to compute a distribution. 
For $x\in[0,1]^E$, let $R(x)\subseteq E$ be a random set obtained by including each element $e\in E$ independently with probability $x_e$.
\begin{definition}
	Let $b,c\in[0,1]$ and let $(E,\cI')$ be an independence system.
	A $(b,c)$-balanced CR scheme $\pi$ for $P(\cI')$ is a procedure that,
	for every $x\in bP$ and $A\subseteq E$,
	returns a random set $\pi_x(A)\subseteq A\cap\supp(x)$ and satisfies the following properties:
	\begin{enumerate}
		\item $\pi_x(A)\in\cI'$ with probability $1$ for all $A\subseteq E$, $x\in bP(\cI')$, and
		\item for all $i\in\supp(x)$, $\Pr[i\in\pi_x(R(x))\mid i\in R(x)]\ge c$ for all $x\in bP(\cI')$.
	\end{enumerate}
\end{definition}

\begin{theorem}[Theorem 1.4 \cite{CVZ2014}]
	For any $b\in(0,1]$, there is an optimal $(b,\frac{1-e^{-b}}{b})$-balanced CR scheme for any matroid polytope.
	Moreover, the scheme is efficiently implementable.
\end{theorem}

Recall that $\hat{x}\in\hat{P}(\cI)$ is a optimal solution of $\max_{x\in \hat{P}(\cI)}\sum_{e\in E}w(e)x_e$ and $\hat{x}\in P(\cI_i)$ for all $i\in [\mu]$.
Let $\pi^i_x$ be a $(b,\frac{1-e^{-b}}{b})$-balanced CR scheme for $P(\cI_i)$.
Then we can efficiently implement a procedure that outputs a random set $\bigcap_{i\in[\mu]}\pi^i_{b\hat{x}}(R(b\hat{x}))$.
By the definition of CR scheme, 
$\bigcap_{i\in[\mu]}\pi^i_{b\hat{x}}(R(b\hat{x}))\in\cI$ with probability $1$ and
$\Pr[i\in\pi_{b\hat{x}}(R(b\hat{x}))\mid i\in R(b\hat{x})]\ge \left(\frac{1-e^{-b}}{b}\right)^\mu$ for all $i\in\supp(\hat{x})$.
Thus, we get
\begin{align*}
\min_{k\in[n]}\mathbb{E}\left[f_k\left(\bigcap_{i\in[\mu]}\pi^i_{b\hat{x}}(R(b\hat{x}))\right)\right]
&\ge 
b\cdot\left(\frac{1-e^{-b}}{b}\right)^{\mu}
\cdot\max_{X\in\cI}\min_{k\in[n]}f_k(X).
\end{align*}
By choosing $b=2/\mu$, the coefficient of the right hand side is at least
\begin{align*}
\frac{2}{\mu}\cdot\left(\frac{1-e^{-2/\mu}}{2/\mu}\right)^{\mu}
\ge \frac{2}{\mu}\cdot\lim_{t\to\infty}\left(\frac{1-e^{-2/t}}{2/t}\right)^{t}
=\frac{2}{e\mu}
\end{align*}
since $\left(\frac{1-e^{-2/t}}{2/t}\right)^{t}$ is monotone increasing for $t>0$.
Hence, $\bigcap_{i\in[\mu]}\pi^i_{b\hat{x}}(R(b\hat{x}))$ is a $2/(e\mu)$-approximate solution.

\subsection{Linear functions in a polytope}\label{sec:LP-based polytope}
We consider the following variant of \eqref{eq:problem}. 
Instead of $n$ functions $f_1, \ldots, f_n$, we are given a set of functions
\[
\cF=\left\{f\,\middle|\, f(X)=\sum_{e\in X}w_e~(\forall X\in 2^E),~Aw+B\psi\le c,\ w\ge 0,\ \psi\ge 0\right\}
\]
for some $A\in\mathbb{R}^{m\times |E|}$, $B\in\mathbb{R}^{m\times d}$, and $c\in\mathbb{R}^m$.
Now, we aim to solve 
\begin{align}\label{eq:problem polytope}
	\max_{p\in\Delta(\cI)} \min_{f\in\cF} \sum_{X\in\cI} p_X f(X). 
\end{align}
Note that for linear functions $f_1,\dots,f_n$, \eqref{eq:problem} is equivalent to \eqref{eq:problem polytope} in which
\begin{align*}
	\cF=
	\conv\{f_1,\dots,f_n\}
	=\left\{~f~\,\middle|\,
	\begin{array}{l}
		f(X)=\sum_{e\in X}w_e~(\forall X\in 2^E),\\
		w_e=\sum_{k\in[n]}q_kf_k(\{e\}),\\
		\sum_{k\in[n]}q_k=1,\ w,q\ge 0
	\end{array}
	\right\}.
\end{align*}

We observe that \eqref{eq:problem polytope} is equal to $\max_{x\in P(\cI)}\min\{x^\top w \mid Aw+B\psi\le c,\ w\ge 0,\ \psi\ge 0\}$ by using a similar argument to Lemma \ref{lem:LP}. 
The LP duality implies that $\min\{x^\top w \mid Aw+B\psi\le c,\ w\ge 0,\ \psi\ge 0\} = \max\{c^\top y \mid A^\top y\ge x,\ B^\top y\ge 0,\ y\ge 0 \}$. 
Thus the optimal value of \eqref{eq:problem polytope} is equal to that of the LP
\(
\max_{x\in P(\cI)}\max_{y:\,A^\top y=x,\ B^\top y\ge 0,\ y\ge 0}\sum_{i\in [m]}b_i y_i.
\)
Hence, Theorems \ref{thm:GLS1} and \ref{thm:GLS2} imply that if the separation problem for $P(\cI)$ can be solved in polynomial time, 
then we can solve \eqref{eq:problem polytope} in polynomial time.

\subsection{Approximately linear case}\label{sec:LP-based approx linear}
In this subsection, we consider the case where the possible functions are approximately linear.
For a function $f\colon 2^E\to\mathbb{R}_+$, we call that $\hat{f}$ is an $\alpha$-approximate of $f$ if
\begin{align*}
\alpha\cdot \hat{f}(X)\le f(X)\le \hat{f}(X) \qquad(\forall X\in 2^E).
\end{align*}
Suppose that $\hat{f}_1,\dots,\hat{f}_n$ are $\alpha$-approximate of $f_1,\dots,f_n$, respectively.
Then, the optimal solution for $\hat{f}_1,\dots,\hat{f}_n$
is an $\alpha$-approximate solution of the original problem.

\begin{theorem}
	If $f_1,\dots,f_n$ have $\alpha$-approximate linear functions $\hat{f}_1,\dots,\hat{f}_n$ and
	there is a polynomial time algorithm to solve the separation problem for $P(\cI)$,
	then there is a polynomial time algorithm to compute an $\alpha$-approximate solution for the robust optimization problem.
\end{theorem}
\begin{proof}
	We can compute optimal solution \(p^*\in\argmax_{p\in\Delta(\cI)}\min_{k\in[n]}\sum_{X\in\cI}p_X\cdot \hat{f}_k(X)\) 
	by Theorem~\ref{thm:LP-based}.
	Then we have
	\begin{align*}
	\min_{k\in[n]}\sum_{X\in\cI}p_X^*\cdot f_k(X)
	&\le \min_{k\in[n]}\sum_{X\in\cI}p_X^*\cdot \hat{f}_k(X)
	=\max_{p\in\Delta(\cI)}\min_{k\in[n]}\sum_{X\in\cI}p_X\cdot \hat{f}_k(X)\\
	&\le \frac{1}{\alpha}\cdot\max_{p\in\Delta(\cI)}\min_{k\in[n]}\sum_{X\in\cI}p_X\cdot f_k(X)
	\end{align*}
	and hence $p^*$ is an $\alpha$-approximate solution for the robust optimization problem.
\end{proof}

Finally, let us see that a monotone submodular function with a small curvature has a linear function with a small approximation factor.
For a monotone submodular function $g\colon 2^E\to\mathbb{R}_+$,
the (total) \emph{curvature} of $g$ is defined as \(c_g=1-\min_{e\in E}\frac{g(E)-g(E\setminus\{e\})}{g(\{e\})}\)~\cite{CC1984}.
Let $\hat{g}(X)=\sum_{e\in X}g(\{e\})$.
Then, we have \((1-c_f)\hat{g}(X)\le g(X)\le \hat{g}(X)\) and hence $g$ is a $(1-c_g)$-approximate of $\hat{g}$.

\section{MWU-based Algorithm}\label{sec:MWU}
In this section, we present an algorithm based on the MWU method~\cite{AHK2012}. 
This algorithm is applicable to general cases. 
We assume that $f_k(X)\ge 0$ for any $k\in[n]$ and $X\in\cI$. 

We describe the idea of our algorithm. 
Recall the minimax relation \eqref{eq:minimax}. 
Let us focus on the right hand side. 
We define weights $\omega_k$ for each function $f_k$, and iteratively update them. 
Intuitively, a function with a larger weight is likely to be chosen with higher probability. 
At the first round, all functions have the same weights. 
At each round $t$, we set a probability $q_k \ (k \in [n])$ that $f_k$ is chosen by normalizing the weights. 
Then we compute an (approximate) optimal solution $X^{(t)}$ of $\max_{X\in\cI}\sum_{k\in[n]} q_k\cdot f_k(X)$. 
To minimize the right hand side of \eqref{eq:minimax}, the probability $q_k$ for a function $f_k$ with a larger value $f_k(X)$ should be decreased. 
Thus we update the weights according to $f_k(X)$. 
We repeat this procedure, and set a randomized strategy $p \in \Delta(\cI)$ according to $X^{(t)}$'s. 

Krause et al.~\cite{krause2011} and Chen et al.~\cite{CLSS2017} proposed the above algorithm when $f_1,\dots,f_n$ are functions with range $[0,1]$.
They proved that if there exists an $\alpha$-approximation algorithm to $\max_{X\in\cI}\sum_{k\in[n]}q_kf_k(X)$ for any $q\in\Delta_n$, then the approximation ratio is $\alpha-\epsilon$ for any fixed constant $\epsilon > 0$.
This implies an approximation ratio of $\alpha-\epsilon\cdot\max_{k\in[n], \, X\in\cI}f_k(X)/\nu^*$ when $f_1,\dots,f_n$ are functions with range $\mathbb{R}_+$, where $\nu^*$ is the optimal value of \eqref{eq:problem}.
Here, $\max_{k\in[n],\,X\in\cI}f_k(X)/\nu^*$ could be large in general.
To remove this term from the approximation ratio, we introduce a novel concept of function transformation. 
We improve the existing algorithms~\cite{CLSS2017, krause2011} with this concept, and show a stronger result later in Theorem~\ref{thm:mwu_main}.
\begin{definition}\label{def:reduction}
For positive reals $\eta$ and $\gamma~(\le 1)$,
we call a function $g$ is an \emph{$(\eta,\gamma)$-reduction} of $f$ if
(i) $g(X)\le \min\{f(X),\,\eta\}$ and
(ii) $g(X)\le \gamma\cdot\eta$ implies $g(X)=f(X)$. 
\end{definition}

We fix a parameter $\gamma>0$, where $1/\gamma$ is bounded by polynomial.
The smaller $\gamma$ is, the wider the class of $(\eta,\gamma)$-reduction of $f$ is. 
We set another parameter $\eta$ later. 
We denote $(\eta,\gamma)$-reduction of $f_1,\dots,f_n$ by $f_1^\eta,\dots,f_n^\eta$, respectively.
In what follows, suppose that we have an $\alpha$-approximation algorithm to
\begin{align}\label{eq:MWU oracle}
\max_{X\in\cI}\sum_{k\in[n]}q_kf_k^\eta(X)
\end{align}
for any $q\in\Delta_n$ and $\eta\in\mathbb{R}_{+}\cup\{\infty\}$. 
In our proposed algorithm, we use $f_k^\eta$ instead of the original $f_k$.
The smaller $\eta$ is, the faster our algorithm converges.
However, the limit outcome of our algorithm as $T$ goes to infinity moves over a little from the optimal solution.
We overcome this issue by setting $\eta$ an appropriate value. 

Our algorithm is summarized in Algorithm \ref{alg:mwu}. 
\begin{algorithm}
\caption{MWU algorithm for the robust optimization}\label{alg:mwu}
\SetKwInOut{Input}{input}
\SetKwInOut{Output}{output}
\Input{positive reals $\eta$, $\delta\le 1/2$, and an integer $T$}
\Output{randomized strategy $p^*\in\Delta(\cI)$}
Let $\omega_k^{(1)}\ot 1$ for each $k\in[n]$\;
\For{$t=1$ to $T$}{
$q_k^{(t)}\ot \omega_k^{(t)}/\sum_{k\in[n]} \omega_k^{(t)}$ for each $k\in[n]$\;
let $X^{(t)}$ be an $\alpha$-approximate solution of $\max_{X\in\cI}\sum_{k\in[n]} q_k^{(t)}\cdot f_k^\eta(X)$ \label{alg:MWU oracle} \;
$\omega_k^{(t+1)}\ot \omega_k^{(t)}(1-\delta)^{f_k^\eta(X^{(t)})/\eta}$ for each $k\in[n]$\;
}
\Return $p^*\in\Delta(\cI)$ such that $p_X^*=|\{t\in\{1,\dots,T\}\mid X^{(t)}=X\}|/T$\;
\end{algorithm}
Note that $f_k^\infty=f_k$ $(k\in[n])$. 
We remark that when the parameter $\gamma$ is small, there may exist a better approximation algorithm for \eqref{eq:MWU oracle}, but the running time of Algorithm \ref{alg:mwu} becomes longer. 

The main result of this section is stated below. 
\begin{theorem}\label{thm:mwu_main}
If there exists an $\alpha$-approximation algorithm to solve \eqref{eq:MWU oracle} for any $q\in\Delta_n$ and $\eta\in\mathbb{R}_+\cup\{\infty\}$,
then Algorithm~\ref{alg:mwu} is an $(\alpha-\epsilon)$-approximation algorithm to the robust optimization problem \eqref{eq:problem} for any fixed $\epsilon>0$.
In addition, the running time of Algorithm~\ref{alg:mwu} is $O(\frac{n^2\ln n}{\alpha\epsilon^3\gamma}\theta)$,
where $\theta$ is the running time of the $\alpha$-approximation algorithm to \eqref{eq:MWU oracle}. 
\end{theorem}

To show this, we prove the following lemma by standard analysis of the multiplicative weights update method (see, e.g., \cite{AHK2012}).
In the following, we denote by $\nu^*$ the optimal value of \eqref{eq:problem}. 
\begin{lemma}\label{lem:mwu}
	For any $\delta\in (0,1/2]$, it holds that 
	\begin{align*}
	\sum_{t=1}^T\sum_{k\in[n]} q_k^{(t)}\cdot f_k^\eta(X^{(t)})\le \frac{\eta\ln n}{\delta}+(1+\delta)\cdot\min_{k\in[n]}\sum_{t=1}^T f_k^\eta(X^{(t)}).
	\end{align*}
\end{lemma}
\begin{proof}
	Let $\Phi^{(t)}=\sum_{k\in[n]} \omega_k^{(t)}$.
	Then we have
	\begin{align*}
	\Phi^{(t+1)}
	&=\sum_{k\in[n]} \omega_k^{(t+1)}=\sum_{k\in[n]} \omega_k^{(t)}(1-\delta)^{f_k^\eta(X^{(t)})/\eta}\\
	&\le \sum_{k\in[n]} \omega_k^{(t)}(1-\delta\cdot f_k^\eta(X^{(t)})/\eta)
	= \Phi^{(t)}\cdot \left(1-\delta \sum_{k\in[n]} q_k^{(t)}\cdot f_k^\eta(X^{(t)})/\eta\right)\\
	&\le \Phi^{(t)}\cdot \exp\left(-\delta \sum_{k\in[n]} q_k^{(t)}\cdot f_k^\eta(X^{(t)})/\eta\right).
	\end{align*}
	Here, the first inequality follows because $(1-\delta)^x\le 1-\delta x$ holds for any $x\in [0,1]$ by convexity of $(1-\delta)^x$, 
	and $f_k^\eta(X^{(t)})/\eta\in [0,1]$ holds by Definition~\ref{def:reduction}(i). 
	The last inequality holds since $1-x\le e^{-x}$ for any $x$.
	Thus, we get
	\begin{align*}
	\Phi^{(T+1)}
	&\le \Phi^{(1)}\cdot \exp\left(-\delta \sum_{t=1}^T\sum_{k\in[n]} q_k^{(t)}\cdot f_k^\eta(X^{(t)})/\eta\right) \\
	&=n\cdot \exp\left(-\delta \sum_{t=1}^T\sum_{k\in[n]} q_k^{(t)}\cdot f_k^\eta(X^{(t)})/\eta\right).
	\end{align*}
	In addition, $\Phi^{(T+1)}\ge \omega_k^{(T+1)}=(1-\delta)^{\sum_{t=1}^T f_k^\eta(X^{(t)})/\eta}$ for each $i\in[n]$.
	Hence, we have
	\begin{align*}
	n\cdot \exp\left(-\delta \sum_{t=1}^T \sum_{k\in[n]} q_k^{(t)}f_k^\eta(X^{(t)})/\eta\right)
	\ge (1-\delta)^{\sum_{t=1}^T f_i^\eta(X^{(t)})/\eta} \quad (\forall i \in [n]).
	\end{align*}
	This implies that
	\begin{align*}
	\sum_{t=1}^T \sum_{k\in[n]} q_k^{(t)}\cdot f_k^\eta(X^{(t)})
	&\le \frac{\eta\ln n}{\delta}+\frac{1}{\delta}\ln\frac{1}{1-\delta}\cdot \sum_{t=1}^T q_i^{(t)}\cdot f_i^\eta(X^{(t)}) \\
	&\le \frac{\eta\ln n}{\delta}+(1+\delta)\cdot \sum_{t=1}^Tq_i^{(t)}\cdot f_i^\eta(X^{(t)})
	\end{align*}
	holds for any $i\in[n]$, because $\frac{1}{\delta}\ln\frac{1}{1-\delta}\le 1+\delta$ holds for any $\delta\in (0,1/2]$. 
	By taking a minimum of the right hand side, we prove the lemma.
\end{proof}

Next, we see that the optimal value of \eqref{eq:problem} for $f_1,\dots,f_n$ and the one for $f^\eta_1,\dots,f^\eta_n$ are close if $\eta$ is a large number.
\begin{lemma}\label{lem:mwu1}
	If $\eta\ge \frac{n}{\delta\gamma}\cdot \nu^*$, we have
	\begin{align*}
	\nu^*\ge \min_{q\in\Delta_n}\max_{X\in\cI}\sum_{k\in[n]} q_k\cdot f_k^\eta(X)\ge (1-\delta)\nu^*.
	\end{align*}
\end{lemma}
\begin{proof}
	The first inequality follows immediately because $f_k(X)\ge f_k^\eta(X)$ for any $k\in[n]$ and $X\in\cI$ by Definition~\ref{def:reduction}(i). 
	Let $q^*\in \argmin_{q\in\Delta_n}\max_{X\in\cI}\sum_{k\in[n]} q_k\cdot f_k^\eta(X)$.
	We denote $J=\{k\in[n]\mid \max_{X\in\cI}f_k(X)< \gamma\eta\}$. 
	Then, for $i\not\in J$, it holds that $q_i^*\le \delta/n$, because
	\begin{align*}
	q_i^*\cdot (n/\delta)\cdot \nu^*
	\le q_i^*\cdot\gamma\eta
	\le \max_{X\in\cI}q_i^*\cdot f^\eta_i(X)
	\le \max_{X\in\cI}\sum_{k\in[n]} q_k^*\cdot f^\eta_k(X)
	\le \nu^*
	\end{align*}
	by Definition~\ref{def:reduction}(ii). 
	Hence, we have $\sum_{j\in J}q_j^*= 1-\sum_{j\not\in J}q_j^*\ge 1-\delta$.
	Then we see that 
	\begin{align*}
	\nu^*
	&=\min_{q\in\Delta_n}\max_{X\in\cI}\sum_{k\in[n]} q_k\cdot f_k(X)\\
	&\le \max_{X\in\cI}\sum_{k\in J} \frac{q_k^*}{\sum_{j\in J}q_j^*}\cdot f_k(X)\\
	&\le \frac{1}{1-\delta}\cdot\max_{X\in\cI}\sum_{k\in J} q_k^*\cdot f_k(X)
	=\frac{1}{1-\delta}\cdot\max_{X\in\cI}\sum_{k\in J} q_k^*\cdot f_k^\eta(X)\\
	&\le \frac{1}{1-\delta}\cdot\max_{X\in\cI}\sum_{k\in [n]} q_k^*\cdot f_k^\eta(X)
	\end{align*}
	and hence $\min_{q\in\Delta_n}\max_{X\in\cI}\sum_{k\in[n]} q_k\cdot f_k^\eta(X)\ge (1-\delta)\nu^*$.
\end{proof}

\begin{lemma}\label{lem:mwu2}
	For any fixed $\epsilon > 0$, the output $p^*$ of Algorithm \ref{alg:mwu} is an $(\alpha-\epsilon)$-approximate solution of \eqref{eq:problem} 
	when we set $T=\lceil \frac{n^2\ln n}{\alpha\delta^3\gamma}\rceil$,
	$\frac{n^2}{\alpha\delta\gamma}\nu^*\ge \eta\ge \frac{n}{\delta\gamma}\nu^*$,
	and $\delta=\min\{\epsilon/3,\,1/2\}$.
\end{lemma}
\begin{proof}
	For any $t\in[T]$, since $X^{(t)}$ is an $\alpha$-approximate solution of $\max_{X\in\cI}\sum_{k\in[n]}q_k^{(t)}\cdot f_k^\eta(X)$, we have
	\begin{align*}  
	\sum_{k\in[n]} q_i^{(t)}\cdot f_k^\eta(X^{(t)})
	&\ge \alpha\cdot \max_{X\in\cI}\sum_{k\in[n]} q_k^{(t)}\cdot f_k^\eta(X) \\
	&\ge \alpha\cdot \min_{q\in\Delta_n}\max_{X\in\cI}\sum_{k\in[n]} q_k\cdot f_k^\eta(X)
	\ge \alpha(1-\delta)\cdot \nu^*
	\end{align*}
	where the last inequality follows from Lemma~\ref{lem:mwu1}. 
	By taking the average over $t$ for both sides, it holds that 
	\begin{align*}
	\alpha(1-\delta)\cdot\nu^*
	&\le \frac{1}{T}\sum_{t=1}^T \sum_{k\in[n]} q_k^{(t)}\cdot f_k^\eta(X^{(t)}). 
	\end{align*}
	Moreover, Lemma~\ref{lem:mwu} implies that 
	\begin{align*}
	\frac{1}{T}\sum_{t=1}^T \sum_{k\in[n]} q_k^{(t)}\cdot f_k^\eta(X^{(t)})&\le \frac{\eta\ln n}{\delta T}+(1+\delta)\cdot \min_{k\in[n]}\sum_{t=1}^T \frac{1}{T}\cdot f_k^\eta(X^{(t)}). 
	\end{align*}
	By the definitions of $p^*$ and $f_k^\eta$, we observe that $\sum_{t=1}^T \frac{1}{T}\cdot f_k^\eta(X^{(t)}) \leq \sum_{X \in \cI} p^*_X f_k^\eta(X) \leq \sum_{X\in\cI}p_X^*\cdot f_k(X)$ for any $k\in[n]$. 
	We see that 
	\begin{align*}
	\alpha(1-\delta)\cdot\nu^*&\le \frac{\eta\ln n}{\delta T}+(1+\delta)\cdot \min_{k\in[n]}\sum_{X\in\cI}p_X^*\cdot f_k(X) \\
	&\le \min_{k\in [n]}\sum_{X\in\cI}p_X^*\cdot f_k(X)+\frac{\eta\ln n}{\delta T}+\delta\nu^*,
	\end{align*}
	where the last inequality holds because 
	\[
	\min_{k\in[n]}\sum_{X\in\cI}p_X^*\cdot f_k(X) \leq \max_{p\in\Delta(\cI)}\min_{k\in[n]}\sum_{X\in\cI}p_X\cdot f_k(X) = \nu^*.
	\] 
	Thus it follows that 
	\begin{align*}
	\min_{k\in[n]}\sum_{X\in\cI}p_X^*\cdot f_k(X)
	&\ge \alpha\cdot \nu^* -\left(\alpha \delta\nu^* + \frac{\eta\ln n}{\delta T} +\delta\nu^*\right) \\
	&\ge \alpha\cdot \nu^* - 3\delta\nu^* \geq (\alpha-\epsilon)\cdot\nu^*,
	\end{align*}
	since we set $T=\lceil \frac{n^2\ln n}{\alpha\delta^3\gamma}\rceil$ and $\delta=\min\{\epsilon/3,\,1/2\}$. 
	Therefore, $p^*$ is an $(\alpha-\epsilon)$-approximate solution.
\end{proof}

We are ready to prove Theorem \ref{thm:mwu_main}. 
\begin{proof}[Proof of Theorem \ref{thm:mwu_main}.]
	We show that we can set the parameter $\eta$ to use Lemma \ref{lem:mwu2}. 
	For each $k\in[n]$, let $X_k'$ be an $\alpha$-approximate solution to \eqref{eq:MWU oracle} with $q_k =1$ and $q_{k'}=0 \ (k'\neq k)$, namely, $\max_{X\in\cI}f_k(X)$.
	Then, by Proposition~\ref{prop:nustar}, we have $\min_{k\in[n]}f_k(X_k')/n\le \nu^*\le \min_{k\in[n]}f_k(X_k')/\alpha$.
	Hence, we obtain $\frac{n^2}{\alpha\delta\gamma}\nu^*\ge \eta\ge \frac{n}{\delta\gamma}\nu^*$
	by setting $\eta=\frac{n}{\alpha\delta\gamma}\cdot\min_{k\in[n]}f_k(X_k')$.
	Thus, the statement follows from Lemma \ref{lem:mwu2}. 
\end{proof}

\begin{remark}
	The output of our algorithm have a support of size at most $T = \lceil \frac{n^2\ln n}{\alpha\delta^3\gamma}\rceil$. 
	Without loss of the objective value, we can find a sparse solution by the following procedure.
	Let $\Delta_T$ be the subset of distributions $\Delta(\cI)$ whose support is a subset of $\{X^{(1)},\dots,X^{(T)}\}$.
	Then, we can obtain the best distribution in $\Delta_T$
	by solving the following LP: 
	\begin{align*}
	\begin{array}{rll}
	\max & \nu&\\
	\text{s.t.}& \nu\le \sum_{t=1}^T f_i(X^{(t)})r_t & (\forall i\in[n]),\\
	& \sum_{t=1}^T r_t=1,&\\
	& r_t\ge 0&(\forall t\in[T]).
	\end{array}
	\end{align*}
	This LP has a polynomial size. 
	If we pick an extreme optimal point $(r,\nu)$, then the support size of $r$ is at most $n$.
\end{remark}

In the subsequent subsections, we enumerate applications of Theorem~\ref{thm:mwu_main}.

\subsection{Robust monotone submodular maximization}\label{sec:mwu_monotone submodular}
Let us consider the case where $f_1,\dots,f_n\colon 2^E\to\mathbb{R}_+$ are monotone submodular functions.
In this case, it requires exponential number of queries to get a $(1-1/e+\epsilon)$-approximation for any $\epsilon>0$
even if $n=1$ and the constraint is a uniform matroid~\cite{NW1978}. 

We set $f_k^\eta(X)=\min\{f_k(X),\,\eta\}$.
Then, $f_k^\eta$ is an $(\eta,1)$-reduction of $f_k$ and $f_k^\eta$ is a monotone submodular function~\cite{lovasz1983,fujito2000}.
Moreover, for any $q\in\Delta_n$, a function $\sum_{k\in[n]}q_kf_k^\eta(X)$ is monotone submodular, 
since a nonnegative linear combination of monotone submodular functions is also monotone submodular.
Thus, \eqref{eq:MWU oracle} is an instance of the monotone submodular function maximization problem. 
There exists $(1-1/e)$-approximation algorithms for this problem under a knapsack constraint~\cite{Sviridenko04} or under a matroid constraint~\cite{CCPV2007,FW2012}. 
When $\cI$ is defined from a knapsack constraint or a matroid, we can see from Theorem~\ref{thm:mwu_main} that Algorithm~\ref{alg:mwu} using these existing algorithms in line \ref{alg:MWU oracle} finds a $(1-1/e-\epsilon)$-approximate solution to \eqref{eq:problem}. 
\begin{theorem}\label{thm:monotone_submodular}
	For any positive real $\epsilon>0$,
	there exists a $(1-1/e-\epsilon)$-approximation algorithm for the robust optimization problem \eqref{eq:problem}
	when $f_1,\dots,f_n$ are monotone submodular and $\cI$ is defined from a knapsack constraint or a matroid.
\end{theorem}

For the monotone submodular maximization subject to $\mu$-matroid intersection,
a $1/(\mu+\epsilon)$-approximation algorithm is known for any fixed $\epsilon>0$~\cite{LSV2010}.
Thus, we obtain the following consequence.
\begin{theorem}\label{thm:monotone_submodular_mu}
	For any fixed positive real $\epsilon>0$,
	there exists a $1/(\mu+\epsilon)$-approximation algorithm for the robust optimization problem \eqref{eq:problem} 
	when $f_1,\dots,f_n$ are monotone submodular and $\cI$ is given by a $\mu$-matroid intersection.
\end{theorem}

A monotone linear maximization subject to $\mu$-matroid intersection can be viewed as a weighted rank-function (which is monotone submodular) maximization subject to $(\mu-1)$-matroid intersection.
Hence, we also obtain the following theorem.
\begin{theorem}\label{thm:monotone_linear_mu}
	For any fixed positive real $\epsilon>0$,
	there exists a $1/(\mu-1+\epsilon)$ for the robust optimization problem \eqref{eq:problem}
	when $f_1,\dots,f_n$ are monotone linear and $\cI$ is given by a $\mu$-matroid intersection.
\end{theorem}

\subsection{Robust (non-monotone) submodular maximization}\label{sec:mwu_submodular}
For each $k\in[n]$, let $f_k\colon 2^E\to\mathbb{R}_+$ be a submodular function with $f_k(\emptyset) = 0$ that is potentially non-monotone. 
Let $\cI=2^E$.
For a (non-monotone) submodular function maximization without constraint, 
there exists a $1/2$-approximation algorithm and it is best possible~\cite{FMVV2011,BFNS2015}.

In this case, $\min\{f_k,\eta\}$ may not a submodular function.
Hence, for each $k\in[n]$, we define
\begin{align*}
f_k^\eta(X)=\min\{f(Z)+\eta\cdot |X-Z|/|E|\mid Z\subseteq X\}.
\end{align*}
Note that $f_k^\eta(X)$ is a submodular function~\cite{fujishige2005} and
we can evaluate the value $f_k^\eta(X)$ in strongly polynomial time by a submodular function minimization algorithm~\cite{schrijver2000,IFF2001}.
Then, we observe that each $f_k^\eta$ is an $(\eta,1/|E|)$-reduction of $f_k$, because
(i) $f_k^\eta(X)\le f_k(X)$ and $f_k^\eta(X)\le f(\emptyset)+\eta\cdot |X|/|E|\le \eta$, and
(ii) $f_k(X)\le \eta/|E|$ implies $f_k^\eta(X)=\min\{f(Z)+\eta\cdot |X-Z|/|E|\mid Z\subseteq X\}=f_k(X)$.

\begin{theorem}\label{thm:non-monotone_submodular}
	For any fixed positive real $\epsilon>0$,
	there exists a $(1/2-\epsilon)$-approximation algorithm for the robust optimization problem \eqref{eq:problem} 
	when $f_1,\dots,f_n$ are submodular and $\cI=2^E$.
\end{theorem}

\subsection{Robust linear maximization}\label{sec:mwu_line}
Let us consider the case where $f_1,\dots,f_n$ are monotone linear, i.e.,
$f_k(X)=\sum_{e\in X}w_{ke}$ where $w_{ke}\ge 0$ for each $k\in[n]$ and $e\in E$. 
For each $k\in[n]$, we define
\begin{align*}
f_k^\eta(X)=\sum_{e\in X}\min\{w_{ke},\,\eta/|E|\}.
\end{align*}
Then, each $f_k^\eta$ is an $(\eta,1/|E|)$-reduction of $f_k$, because
(i) $f_k^\eta(X)\le f_k(X)$ and $f_k^\eta(X)\le \eta\cdot |X|/|E|\le \eta$, and
(ii) $f_k(X)\le \eta/|E|$ implies $f_k^\eta(X)=\sum_{e\in X}\min\{w_{ke},\,\eta/|E|\}=f_k(X)$.
Then, \eqref{eq:MWU oracle} can be rewritten as
\begin{align*}
\max_{X\in\cI}\sum_{k\in[n]}q_kf_k^\eta(X)=\max_{X\in\cI}\sum_{e\in X}\left(\sum_{k\in[n]}q_k\min\{w_{ke},\,\eta/|E|\}\right). 
\end{align*}

We focus on the case $\cI$ is defined from a knapsack constraint, i.e., $\cI=\{X\subseteq E\mid \sum_{e\in X}s(e)\le C\}$, where $s(e)$ is size of each item $e \in E$ and $C$ is a capacity. 
In this case, \eqref{eq:MWU oracle} is a knapsack problem instance. 
It is known that the knapsack problem admits an FPTAS~\cite{KMS2000}. 
Hence, by applying Theorem \ref{thm:mwu_main}, we obtain the following result.

\begin{theorem}\label{thm:monotone_linear_knapsack}
	There exists an FPTAS for the robust optimization problem \eqref{eq:problem}
	when $f_1,\dots,f_n$ are monotone linear and $\cI$ is defined from a knapsack constraint.
\end{theorem}

\subsection{Cardinality Robustness for the Knapsack Problem}\label{sec:mwu_cardinality}
Finally, we apply Theorem \ref{thm:mwu_main} to the problem of maximizing the cardinality robustness for the knapsack problem.
We are given a set of items $E=\{1,2,\dots,n\}$ with size $s:E\to\mathbb{R}_{++}$, and value $v:E\to\mathbb{R}_{++}$, and a knapsack capacity $C$. 
Without loss of generality, we assume $s(e)\le C$ for each $e\in E$. 
The set of feasible solutions is $\cI=\{X\subseteq E\mid s(X)\le C\}$.

Let us denote
\begin{align*}
v_{\le k}(X)=\max\{\textstyle\sum_{e\in X'}v(e)\mid |X'|\le k,~X'\subseteq X\}
\end{align*}
and let $X_k^*\in\argmax_{X\in\cI}v_{\le k}(X)$.
The \emph{cardinality robustness} of $X \in \cI$ is defined as
\begin{align*}
\min_{k\in[n]}\frac{v_{\le k}(X)}{v_{\le k}(X_k^*)}.
\end{align*}
The \emph{maximum cardinality robustness problem} for the knapsack problem is to find a randomized solution $p \in \Delta(\cI)$ with the maximum cardinality robustness under the knapsack constraint. 
We show that this problem is NP-hard (see Theorem~\ref{thm:cknapsack_hard} in Section~\ref{sec:cardinality NP-hard}) but admits an FPTAS. 
Note that Kakimura et al.~\cite{Kakimura2012} proved that its deterministic version is NP-hard but also admits an FPTAS. 

We can see the maximum cardinality robustness problem as the robust optimization problem \eqref{eq:problem} with constraint $\cI$ and objective functions
\begin{align*}
  f_k(X)=\frac{v_{\le k}(X)}{v_{\le k}(X_k^*)}\qquad(k\in[n]).
\end{align*}
We remark that although the evaluation of $f_k(X)$ for a given solution $X$ is NP-hard due to $v_{\le k}(X_k^*)$, it holds that $\max_{X\in\cI}f_k(X)=f_k(X_k^*)=1$ for each $k\in [n]$.  
Thus, when we choose $\eta \ge 1$ and set $f_k^\eta=f_k$, we observe that $f_k^\eta$ is a $(\eta, 1)$-reduction of $f_k$. 
To use Algorithm~\ref{alg:mwu}, we need to construct an approximation algorithm to \eqref{eq:MWU oracle} for any $q\in\Delta_n$.
We provide an FPTAS for this, which implies that Algorithm~\ref{alg:mwu} is an FPTAS to the maximum cardinality robustness problem. 

In our FPTAS for \eqref{eq:MWU oracle}, we use the following FPTAS as a subroutine.
\begin{lemma}[\cite{CapraraKPP2000}]\label{lem:ck_FPTAS}
  There exists an FPTAS to compute the value of $v_{\le k}(X_k^*)$ for each $k\in [n]$.
\end{lemma}

Then we construct our FPTAS based on the dynamic programming.
\begin{lemma}\label{lem:crk_FPTAS}
  Given $q\in\Delta_n$, there exists an FPTAS to solve \eqref{eq:MWU oracle}. 
\end{lemma}
\begin{proof}
  Let $\epsilon$ be any positive real and 
  let $\nu'=\max_{X\in\cI}\sum_{k\in[n]}\frac{q_k\cdot v_{\le k}(X)}{v_{\le k}(X_k^*)}$ be the optimal value of \eqref{eq:MWU oracle}. 
  For each $k \in[n]$, let $v_k^*$ be any value such that $v_{\le k}(X_k^*)\le v_k^*\le v_{\le k}(X_k^*)/(1-\epsilon)$ $(k\in[n])$.
  We can compute such values in polynomial time with respect to $n$ and $1/\epsilon$ by Lemma~\ref{lem:ck_FPTAS}. 
  For each $X \in \cI$, let $f(X)=\sum_{k\in [n]}\frac{q_k\cdot v_{\le k}(X)}{v_k^*}$.
  Then by the definition of $v_k^*$, it holds that for any $X \in \cI$ 
  \begin{align*}
    \nu'\ge \max_{X\in\cI}f(X)\ge (1-\epsilon)\nu'.
  \end{align*}
  Hence, we aim to solve $\max_{X\in\cI}f(X)$ to know an approximate solution to \eqref{eq:MWU oracle}. 
  
  A simple dynamic programming based algorithm to maximizing $f$ runs in time depending on the value of $f$. 
  By appropriately scaling the value of $f$ according to $\epsilon$, we will be able to obtain a solution whose objective value is at least $(1-2\epsilon)\nu'$ in polynomial time with respect to both $n$ and $1/\epsilon$. 
  Let $\kappa=\lceil\frac{n^2}{\epsilon}\rceil$. 
  For each $X=\{e_1,\dots,e_\ell\}$ with $v(e_1)\ge \dots\ge v(e_\ell)$, we define
  \begin{align*}
    \bar{f}(X)&=\sum_{k=1}^\ell \left\lfloor\left(\sum_{\iota=k}^n\frac{q_\iota}{v^*_\iota}\right)v(e_k)\cdot\kappa\right\rfloor/\kappa.
  \end{align*}
  Note that 
  \begin{align*}
  f(X)
  =\sum_{k=1}^n\frac{q_k\cdot v_{\le k}(X)}{v_k^*}
  =\sum_{k=1}^n \frac{q_k}{v^*_k}\sum_{\iota=1}^{\min\{k,\ell\}}v(e_\iota)
  =\sum_{k=1}^\ell \left(\sum_{\iota=k}^n\frac{q_\iota}{v^*_\iota}\right)v(e_k).
  \end{align*}
  Thus $f(X) \geq \bar{f}(X)$ holds for all $X$. 
  Then we have
  \begin{align*}
  \nu' \geq \max_{X\in\cI}f(X)\ge \max_{X\in\cI}\bar{f}(X). 
  \end{align*}
  We observe that  $\nu'\ge\ \min_{q' \in \Delta_n} \max_{X \in \cI} \sum_{k \in [n]}q'_k f_k^\eta(X)$ and 
  this is at least $\min_{k\in[n]}f_k(X_k^*)/n=1/n$ by $f_k^\eta = f_k \ (k \in [n])$ and Proposition~\ref{prop:nustar}. 
  Thus we have $\nu' \geq 1/n$. 
  Then 
  \begin{align*}
    \max_{X\in\cI}\bar{f}(X)
    &\ge \max_{X\in\cI}f(X)-\frac{n}{\kappa}
    \ge \max_{X\in\cI}f(X)-\frac{\epsilon}{n}
    \ge (1-\epsilon)\nu'-\epsilon\cdot \nu'=(1-2\epsilon)\nu'.
  \end{align*}
  This implies that there exists an FPTAS to solve \eqref{eq:MWU oracle} if we can compute $\max_{X\in\cI}\bar{f}(X)$ in polynomial time in $n$ and $1/\epsilon$.

  Let $\tau(\zeta,\xi,\phi)=\min\{s(X)\mid X\subseteq\{1,\dots,\zeta\},~|X|=\xi,~\kappa\cdot \bar{f}(X)=\phi\}$.
  Here, we assume that $E=\{1,2,\dots,n\}$ and $v(1)\ge v(2)\ge \dots\ge v(n)$. 
  Then we can compute the value of $\tau(\zeta,\xi,\phi)$ by dynamic programming as follows:
  \begin{align*}
    \tau(\zeta,\xi,\phi)=
    \min\biggl\{\tau(\zeta-1,\xi,\phi),s(\zeta)+\tau\biggl(\zeta-1,\xi-1,\phi-\biggl\lfloor\biggl(\sum_{\iota=\zeta}^n\frac{q_\iota}{v^*_\iota}\biggr)v(e_\iota)\cdot\kappa\biggr\rfloor\biggr)\biggr\}. 
  \end{align*}
  We see that $\max_{X\in\cI}\bar{f}(X)=\max\{\phi\mid \tau(n,\xi,\phi)\le C, 0\le \xi\le n\}$. 
  
  It remains to discuss the running time. 
  For all $X$, since $f(X) \leq \sum_{k\in [n]} q_k \leq 1$ and $f(X) \geq \bar{f}(X)$, we see that $\kappa \cdot \bar{f}(X)$ is an integer in $[0, \kappa]$. 
  Hence, there exist $\kappa+1$ possiblities of $\bar{f}(X)$. 
  Therefore, we can compute $\max_{X\in\cI}\bar{f}(X)$ in $O(n^2\kappa)=O(n^4/\epsilon)$ time. 
\end{proof}

Therefore, we can see the following theorem by combining Theorem~\ref{thm:mwu_main} and Lemma~\ref{lem:crk_FPTAS}. 
\begin{theorem}\label{thm:cardinality_knapsack}
	There exists an FPTAS to solve the maximum cardinality robustness problem for the knapsack problem.
\end{theorem}

\subsubsection{NP-hardness of the cardinality robustness for the knapsack problem}\label{sec:cardinality NP-hard} 
We give a reduction from the partition problem with restriction that two partitioned subsets are restricted to have equal cardinality, which is an NP-complete problem \cite{garey1979cai}. 
Given even number of positive integers $a_1,a_2,\dots,a_{2n}$, the problem is to find a subset $I\subseteq [2n]$ such that $|I|=n$ and $\sum_{i\in I}a_i=\sum_{i\in [2n]\setminus I}a_i$. 
Recall that $[2n]=\{1,2,\dots,2n\}$.

\begin{theorem}\label{thm:cknapsack_hard}
	It is NP-hard to find a solution $p\in \Delta(\cI)$ with the maximum cardinality robustness for the knapsack problem.
\end{theorem}
\begin{proof}
	Let $(a_1,a_2,\dots,a_{2n})$ be an instance of the partition problem. 
	Without loss of generality, we assume that $a_1\ge a_2\ge \dots\ge a_{2n}~(\ge 1)$ and $n\ge 4$.
	Define $A=\sum_{i=1}^{2n}a_i/2$.
	We construct the following instance of the maximum cardinality robustness problem for the knapsack problem:
	\begin{itemize}
		\item $E=\{0, 1, \ldots, 2n\}$,
		\item $s(0)=A+2n^2a_1$, $v(0)=2(2n+1)a_1$,
		\item $s(i)=v(i)=a_i+2na_1$ $(i=1,\dots,2n)$,
		\item $C=\sum_{i=1}^{2n} s(i) = 2A+4n^2a_1$.
	\end{itemize}
	Note that $C/2 = s(0) \geq s(1) \geq \dots \geq s(2n)$ and $v(0) = 2v(1) > v(1) \geq \dots \geq v(2n)$. 
	We denote by $\cI$ the set of knapsack solutions. 
	
	Let 
	\begin{align*}
	\alpha= \frac{1}{4} \cdot \frac{3C-2v(0)}{C-v(0)} =  \frac{1}{4}\cdot\frac{3(A+2n^2a_1)-2(2n+1)a_1}{A+2n^2a_1-(2n+1)a_1}.
	\end{align*}
	We claim that this instance has a randomized $\alpha$-robust solution $p \in \Delta(\cI)$ if and only if the partition problem instance has a solution. 
	Recall that $p \in \Delta(\cI)$ is called $\alpha$-robust if $\sum_{X\in\cI} p_X \cdot v_{\le k}(X) \geq \alpha \cdot \max_{Y\in\cI}v_{\le k}(Y)$ for any $k\in[n]$.
	
	Suppose that $I\subseteq [2n]$ is a solution to the partition problem instance, i.e., $|I|=n$ and $\sum_{i\in I}a_i=\sum_{i\in [2n]\setminus I}a_i~(=A)$. 
	Let
	\begin{align*}
	r= \frac{1}{2} \cdot \frac{C-2v(0)}{C-v(0)} =\frac{1}{2}\cdot\frac{A+2n^2a_1-2(2n+1)a_1}{A+2n^2a_1-(2n+1)a_1}.
	\end{align*} 
	Note that $r=2(1-\alpha)$. 
	We define a randomized solution $p \in \Delta(\cI)$ by $p_Y = r$ if $Y=[2n]$, $p_Y = 1-r$ if $Y=\{0\}\cup I$, and $p_Y = 0$ otherwise. 
	Note that $[2n] \in \cI$ by the definition of $C$ and $\{0\} \cup I \in \cI$ because $s(0) = C/2$ and $\sum_{i \in I} s(i) = C/2$. 
	We claim that $p$ is an $\alpha$-robust solution by showing that $v_{\le k}(p)/v_{\le k}(X_k^*) \geq \alpha$ for all $k=1, \ldots, 2n+1$. 
	Recall that $X_k^*\in\argmax_{X\in\cI}v_{\le k}(X)$.

	First, take arbitrarily $k \in \{1, \ldots, n+1\}$. 
	We have $v_{\le k}(X_k^*)\le \sum_{i=0}^{k-1} v(i) \leq (k+1)v(1)$. 
	Moreover, $v_{\le k}([2n]) = \sum_{i=1}^{k}v(i) \geq v(1)+(k-1)2na_1$, and $v_{\le k}(\{0\} \cup I) \geq v(0)+(k-1)2na_1$. 
	Thus we see that 
	\begin{align*}
	\frac{v_{\le k}(p)}{v_{\le k}(X_k^*)}&\ge \frac{(r(2n+1)a_1+(1-r)\cdot 2(2n+1)a_1)+(k-1)\cdot 2na_1}{(k+1)(2n+1)a_1}\\
	&= \frac{2n}{2n+1} + \left(\frac{2}{2n+1}-r\right)\cdot \frac{1}{k+1}. 
	\end{align*}
	Here, because $a_1 \leq A \leq 2na_1$ and $n \geq 4$, it follows that
	\begin{align}\label{eq:alphaupper}
	\alpha
	=\frac{1}{4}\cdot\frac{3(A+2n^2a_1)-2(2n+1)a_1}{A+2n^2a_1-(2n+1)a_1}
	\le \frac{1}{4}\cdot\frac{6n^2-4n+1}{2n(n-1)} 
	\le \frac{2n}{2n+1}. 
	\end{align}
	Then, we see that $\frac{2}{2n+1}-r \leq 0$ since $r = 2(1-\alpha)\ge 2(1-\frac{2n}{2n+1})=\frac{2}{2n+1}$. 
	This implies that 
	\begin{align*}
	\frac{v_{\le k}(p)}{v_{\le k}(X_k^*)} \geq \frac{2n}{2n+1} + \left(\frac{2}{2n+1}-r\right)\cdot \frac{1}{2} 
	= 1 - \frac{r}{2} = \alpha.
	\end{align*}
	
        Next, assume that $k=n+2$. 
        We claim that $v_{\le k}(X_k^*)\le (n+2)\cdot (2n+1)a_1$. 
        To describe this, let $X$ be any set in $\cI$ with $0 \in X$. 
        Since $s(0)=C/2$ and $s(i)=v(i)$ for all $i \geq 1$, we have $\sum_{i \in X} v(i) \leq v(0) + C/2 = v(0) + \sum_{i \in I}v(i) \leq v(0)+\sum_{i\in I} v(1) \leq (n+2)v(1)$. 
        On the other hand, for any set $X \in \cI$ with $0 \not\in X$, we have $v_{\le k}(X_k^*) \leq \sum_{i=1}^{n+2} v(i) \leq (n+2)v(1)$. 
        Hence, $v_{\le k}(X_k^*)\le (n+2)\cdot (2n+1)a_1$. 
        In addition, we observe that $v_{\le k}([2n]) = \sum_{k=1}^{n+2}v(i) \geq (n+2)\cdot 2na_1$ and $v_{\le k}(\{0\} \cup I) \geq v(0)+n\cdot 2na_1 \geq (n+2)2na_1$. 
        Thus, $v_{\le k}(p) \ge (n+2)\cdot 2na_1$.
        These facts together with \eqref{eq:alphaupper} imply that
        \begin{align*}
          \frac{v_{\le k}(p)}{v_{\le k}(X_k^*)}
          &\ge \frac{(n+2)2na_1}{(n+2)(2n+1)a_1}
          =\frac{2n}{2n+1}
          \ge \alpha. 
	\end{align*}

        Let us consider the case when $k \in \{n+3, \ldots, 2n-1\}$.
        For any $X \in \cI$ with $0 \in X$, we observe that $v_{\le k}(X) \leq (n+2)(2n+1)a_1 \leq 2n(n+3)a_1\le \sum_{i=1}^k v(i) = v_{\le k}([2n])$, where the first inequality holds by a similar argument to the one in the above case. 
        Thus, we see that $v_{\le k}(X_k^*)=v_{\le k}([2n])$. 
        Because $v_{\le k}([2n]) \leq v_{\le k+1}([2n])$ and $v_{\le k}(\{0\} \cup I) = v_{\le k+1}(\{0\} \cup I)$, it holds that
        \begin{align*}
          \frac{v_{\le k}(p)}{v_{\le k}(X_k^*)}
          &=\frac{r\cdot v_{\le k}([2n])+(1-r)v_{\le k}(\{0\} \cup I)}{v_{\le k}([2n])}\\
          &=r+\frac{(1-r)v_{\le k}(\{0\}\cup I )}{v_{\le k}([2n])}\\
          &\ge r+\frac{(1-r)v_{\le k}(\{0\} \cup I)}{v_{\le k+1}([2n])}\\
          &=\frac{r\cdot v_{\le k+1}([2n])+(1-r)v_{\le k}(\{0\} \cup I)}{v_{\le k+1}([2n])}\\
          &=\frac{v_{\le k+1}(p)}{v_{\le k+1}(X_{k+1}^*)}.
        \end{align*}
        
        Hence, it remains to show that $\frac{v_{\le k}(p)}{v_{\le k}(X_k^*)} \geq \alpha$ when $k =2n$ and $k=2n+1$. 
        It is clear that $v_{\le 2n}(p) = v_{\le 2n+1}(p)$ since $v_{\le 2n}([2n]) = v_{\le 2n+1}([2n])$ and $v_{\le 2n}(\{0\} \cup I) = v_{\le 2n+1}(\{0\} \cup I)$. 
        We have also $v_{\le 2n+1}(X_k^*) = v_{\le 2n}(X_k^*) = \sum_{i=1}^{2n}v(i) = C$ because $[2n+1] \not\in \cI$. 
        Thus, it follows that
	\begin{align*}
        \frac{v_{\le 2n+1}(p)}{v_{\le 2n+1}(X_k^*)}
        &= \frac{v_{\le 2n}(p)}{v_{\le 2n}(X_{2n}^*)}\\
        &=\frac{rC+(1-r)(v(0)+C/2)}{C} \\
        &=\frac{1+r}{2}+\frac{(1-r)v(0)}{C}\\
        &=1-\frac{r}{2} = \alpha,
	\end{align*}
	where the last second equation holds because $r = (C-2v(0))/2(C-v(0)) \Leftrightarrow (1-r)v(0) = C(1/2-r)$. 
	Therefore, $p$ is $\alpha$-robust.
	
	It remains to prove that if the partition problem instance has no solution, then there exists no $\alpha$-robust solution. 
	Let $p\in\Delta(\cI)$ be a solution and let $r=\sum_{X : 0\not\in X\in\cI} p_X$.
	To show a contradiction, we assume that $p$ is $\alpha$-robust.
	Then it must hold that 
	\begin{align*}
	\frac{v_{\le 1}(p)}{v_{\le 1}(X_1^*)}
	&=\frac{r \cdot v(1)+(1-r)\cdot v(0)}{v(0)} = \frac{(2-r)v(1)}{2v(1)} =\frac{2-r}{2}\\
	&\ge \alpha = \frac{1}{4}\cdot \frac{3C-2v(0)}{C-v(0)}
	\end{align*}
	and hence
	\begin{align*}
	r\le\frac{1}{2}\cdot \frac{C-2v(0)}{C-v(0)}. 
	\end{align*}
	This implies that $r < 1$ and $p_{X} > 0$ for some $X \in \cI$ with $0 \in X$. 
	
	On the other hand, we claim that ${v_{\le 2n+1}(p)}/{v_{\le 2n+1}(X_{2n+1}^*)} < \alpha$. 
	For any $X \in \cI$ with $0 \not\in X$, we observe that $v_{\le 2n+1}(X) \leq \sum_{i=1}^{2n} v(i) =C$. 
	Take an arbitrary set $Y$ such that $Y \cup \{0\} \in \cI$. 
	It holds that $v_{\le 2n+1}(Y\cup \{0\}) = v(0)+\sum_{i \in Y} s(i) \leq v(0)+C/2$.
        We claim that $\sum_{i \in Y} s(i) \ne C/2$.
        If $|Y|>n$, then $\sum_{i \in Y} s(i)\ge (n+1)\cdot 2na_1>na_1+2n^2a_1\ge A+2n^2a_1=C/2$.
        If $|Y|<n$, then $\sum_{i \in Y} s(i)\le (n-1)\cdot (2n+1)a_1\le -n-1+2n^2a_1<A+2n^2a_1=C/2$.
        If $|Y|=n$, then $\sum_{i \in Y} s(i) = \sum_{i \in Y} a_i + 2n^2 a_1 \ne A+2n^2a_1 \, (= C/2)$ since the partition problem instance has no solution.
	Thus we see that $v_{\le 2n+1}(p) < r\cdot C + (1-r)\cdot (v(0)+C/2)$ by $r<1$. 
	Since $v_{\le 2n+1}(X_{2n+1}^*) = C$ and $r\le\frac{1}{2}\cdot \frac{C-2v(0)}{C-v(0)}$, we have
	\begin{align*}
	\frac{v_{\le 2n+1}(p)}{v_{\le 2n+1}(X_{2n+1}^*)}
	&<\frac{rC+(1-r)(v(0)+C/2)}{C}\\
	&=\frac{r(C-2v(0))}{2C} + \frac{C+2v(0)}{2C} \\
	&\leq \frac{1}{4}\cdot \frac{(C-2v(0))^2 + 2(C-v(0))(C+2v(0))}{C(C-v(0))} \\
	&= \frac{1}{4}\cdot \frac{3C-2v(0)}{C-v(0)} = \alpha. 
	\end{align*}
	This implies that $p$ cannot be $\alpha$-robust. 
	
	Therefore, there exists a randomized $\alpha$-robust solution $p \in \Delta(\cI)$ if and only if the partition problem instance has a solution. 
\end{proof}

\section*{Acknowledgments}
The first author is supported by JSPS KAKENHI Grant Number JP16K16005 and JST ACT-I Grant Number JPMJPR17U7.
The second author is supported by JST ERATO Grant Number JPMJER1201, Japan, and JSPS KAKENHI Grant Number JP17K12646.

\bibliography{robust}

\end{document}